\documentclass{amsart}
\usepackage{amsmath, latexsym, amsfonts, stmaryrd, amssymb, amsthm, amscd, array, caption, enumerate}
\usepackage[T1]{fontenc}
\usepackage[toc,page]{appendix}
\usepackage[utf8]{inputenc}
\usepackage{graphics,epsf,psfrag,epsfig,color}
\usepackage[colorlinks=true, pdfstartview=FitV, linkcolor=blue, citecolor=blue, urlcolor=blue,pagebackref=false]{hyperref}

\usepackage[showlabels,sections,floats,textmath,displaymath]{}

\numberwithin{equation}{section}
\newtheorem{theorem}{Theorem}[section]
\newtheorem{lemma}[theorem]{Lemma}
\newtheorem{proposition}[theorem]{Proposition}

\newtheorem{rem}[theorem]{Remark}

\renewenvironment{proof}[1][Proof]{\begin{trivlist}
\item[\hskip \labelsep {\bfseries #1}]}{\qed\end{trivlist}}

\newcommand{\ind}{\mathbf{1}}
\renewcommand{\ge}{\geq}
\renewcommand{\le}{\leq}

\renewcommand{\tilde}{\widetilde}
\renewcommand{\hat}{\widehat}

\DeclareMathSymbol{\leqslant}{\mathalpha}{AMSa}{"36} 
\DeclareMathSymbol{\geqslant}{\mathalpha}{AMSa}{"3E} 
\DeclareMathSymbol{\eset}{\mathalpha}{AMSb}{"3F}     
\renewcommand{\leq}{\;\leqslant\;}                   
\renewcommand{\geq}{\;\geqslant\;}                   
\newcommand{\dd}{\,\text{\rm d}}             
\DeclareMathOperator*{\union}{\bigcup}       

\newcommand{\sumtwo}[2]{\sum_{\substack{#1 \\ #2}}} 
\newcommand{\prodtwo}[2]{\prod_{\substack{#1 \\ #2}}}     

\renewcommand{\u}[1]{\underline{#1}}

\newcommand{\cN}{{\ensuremath{\mathcal N}} }

\newcommand{\cI}{{\ensuremath{\mathcal I}} }

\newcommand{\bP}{{\ensuremath{\mathbf P}} }
\newcommand{\bE}{{\ensuremath{\mathbf E}} }

\newcommand{\bZ}{{\ensuremath{\mathbf Z}} }

\newcommand{\bbE}{{\ensuremath{\mathbb E}} }

\newcommand{\bbN}{{\ensuremath{\mathbb N}} }

\newcommand{\bbP}{{\ensuremath{\mathbb P}} }

\newcommand{\bbR}{{\ensuremath{\mathbb R}} }

\newcommand{\bbZ}{{\ensuremath{\mathbb Z}} }

\newcommand{\ga}{\alpha}
\newcommand{\gb}{\beta}
\newcommand{\gd}{\delta}
\newcommand{\gep}{\varepsilon}       
\newcommand{\gp}{\varphi}

\newcommand{\gD}{\Delta}

\newcommand{\go}{\omega}

\newcommand{\gl}{\lambda}

\makeatletter
\def\captionfont@{\footnotesize}
\def\captionheadfont@{\scshape}

\long\def\@makecaption#1#2{%
  \vspace{2mm}
  \setbox\@tempboxa\vbox{\color@setgroup
    \advance\hsize-6pc\noindent
    \captionfont@\captionheadfont@#1\@xp\@ifnotempty\@xp
        {\@cdr#2\@nil}{.\captionfont@\upshape\enspace#2}%
    \unskip\kern-6pc\par
    \global\setbox\@ne\lastbox\color@endgroup}%
  \ifhbox\@ne 
    \setbox\@ne\hbox{\unhbox\@ne\unskip\unskip\unpenalty\unkern}%
  \fi
  \ifdim\wd\@tempboxa=\z@ 
    \setbox\@ne\hbox to\columnwidth{\hss\kern-6pc\box\@ne\hss}%
  \else 
    \setbox\@ne\vbox{\unvbox\@tempboxa\parskip\z@skip
        \noindent\unhbox\@ne\advance\hsize-6pc\par}%
\fi
  \ifnum\@tempcnta<64 
    \addvspace\abovecaptionskip
    \moveright 3pc\box\@ne
  \else 
    \moveright 3pc\box\@ne
    \nobreak
    \vskip\belowcaptionskip
  \fi
\relax
}
\makeatother
\def\writefig#1 #2 #3 {\rlap{\kern #1 truecm
\raise #2 truecm \hbox{#3}}}

\newcommand{\tf}{\mathtt{F}}

\newcommand{\tg}{\mathtt{G}}

\newcommand{\Var}{\mathbb{V}\mathrm{ar}}

\newcommand{\K}{\mathrm{K}}

\newcommand{\free}{\mathtt{f}}

\renewcommand{\subset}{\subseteq}
\newcommand{\logg}{\log \log }

\newcommand{\m}{\mathtt{m}}
\renewcommand{\phi}{\varphi}

\begin{document}

\title[Disorder relevance for the pinning model]{Pinning on a defect line:\\
characterization of marginal disorder relevance and sharp asymptotics for the critical point shift}

\author{Quentin Berger}
\address{LPMA, Universit\'e Pierre et Marie Curie\\
Campus Jussieu, case 188\\
4 place Jussieu, 75252 Paris Cedex 5, France}
\email{quentin.berger@upmc.fr}

\author{Hubert Lacoin}
\address{IMPA,
Estrada Dona Castorina 110,
Rio de Janeiro / Brasil 22460-320}
\email{lacoin@impa.br}

\begin{abstract}
The effect of disorder for pinning models is a subject which has attracted much attention in theoretical physics and rigorous mathematical physics. A peculiar point of interest is the question of coincidence of the quenched and annealed critical point for a 
small amount of disorder. The question has been mathematically settled in most cases in the last few years, giving in particular 
a rigorous validation of the Harris Criterion on disorder relevance. 
However, the marginal case, where 
the return probability exponent is equal to $1/2$, i.e. where the inter-arrival law of the renewal process is given by 
$\K(n)=n^{-3/2}\phi(n)$
 where $\phi$ is a slowly varying function, has been left partially open. 
In this paper, we give a complete answer to the question by proving 
a simple necessary and sufficient criterion on the return probability for disorder relevance, which confirms earlier 
predictions from the literature.
Moreover, we also provide sharp asymptotics on the critical point shift: 
in the case of the pinning of a one dimensional simple random walk, 
the shift of the critical point satisfies the following high temperature asymptotics
$$
 \lim_{\gb\to 0}\gb^2\log h_c(\gb)= - \frac{\pi}{2}.
$$
This gives a rigorous proof to a claim of
B. Derrida, V. Hakim and J. Vannimenus (Journal of Statistical Physics, 1992).
   \\[10pt]
  2010 \textit{Mathematics Subject Classification: 60K35, 60K37, 82B27, 82B44}
  \\[10pt]
  \textit{Keywords:  Disordered Pinning/Wetting Model, Localization Transition, Disorder Relevance, Harris Criterion.}
\end{abstract}
\maketitle

\section{Introduction}

In statistical mechanics, the introduction of disorder into a system, i.e.~of a random inhomogeneity in the Hamiltonian,
can drastically change its critical behavior. However, this change of behavior does not always occur and in some cases
the disorder system keeps the features of the homogeneous one, at least for small intensities of disorder.
As most systems encountered in nature possesses some kind of microscopic impurities, this question of \textit{disorder relevance}
i.e.\ whether a disordered system behaves like the homogeneous one, has been the object of a lot of attention in the physics community
(see e.g. \cite{cf:Hcrit} and references therein).

\medskip

The present paper deals with the question of influence of disorder for the pinning model. 
This subject has been the object of a lot of studies in the past decades,
either in theoretical physics see e.g.\ \cite{cf:BM1,cf:BM2,cf:CH,cf:DHV,cf:DR,cf:FLNO,cf:GN,cf:KM,cf:KL,cf:Monthus,cf:MG,cf:TC}, or rigorous mathematical physics 
\cite{cf:A06, cf:AZ08, cf:AZ10, cf:AZ12, cf:QB, cf:Magic5, cf:QH, cf:CdH, cf:CTT, cf:CDH, cf:DGLT09, cf:GLT, cf:GLT10, cf:GLT11,
cf:GT_cmp, cf:GT09, cf:Lmart, cf:Poi, cf:T08, cf:T_fractmom}.
A reason why the question of disorder relevance in the special case of the pinning model has focus that much of interest is that 
it is a rather simple framework with a rich phenomenology, and thus gives a good context to test the general prediction made by physicists 
concerning relevance of disorder \cite{cf:Hcrit}. 
Indeed the \emph{pure model} (i.e.\ the one without disorder)
is exactly solvable in the sense that there is an explicit simple expression for the free energy, see \cite{cf:Fisher}, but the specific-heat exponent $\nu_{\mathrm{SH}}$ associated 
to it can take any value in the interval $(-\infty,1]$ by tuning the value of the parameter $\alpha$ introduced in Equation \eqref{alpha}:

\begin{equation}\label{SH}
 \nu_{\mathrm{SH}}=\max(1,2-\alpha^{-1}).
\end{equation}

\medskip

Because of these characteristics, the disordered pinning model has been an ideal candidate to check rigorously the validity
of the renormalization group predictions, and in particular, that of the Harris criterion \cite{cf:Hcrit}.
The principal idea of Harris criterion is that one can predict the effect of a small quantity of disorder by looking
at the properties of the pure system: disorder relevance only depends on the sign of the specific-heat exponent.

\medskip

Specifically, when applied to the pinning model, the criterion leads to the following prediction: 
\begin{itemize}
 \item when the return probability exponent $\alpha$ is strictly larger than $1/2$, then disorder is relevant;
 \item  when $\alpha$ is smaller than $1/2$, disorder is irrelevant;
  \item there is no specific prediction for the case $\alpha=1/2$, where the specific-heat exponent vanishes.
\end{itemize}

Specific studies concerning disordered pinning  \cite{cf:DHV,cf:FLNO} give more detailed predictions:
in the case $\alpha>1/2$, there is a shift of the critical point of the disordered system with respect 
to the annealed one; whereas for $\alpha<1/2$, the two critical points coincide, at least when the inverse temperature $\gb$ is small.
The case $\alpha=1/2$ has been studied in the physics literature but has been the source of some controversy:
in the case where no slowly varying function is present, which corresponds to the classical models of two-dimensional wetting of a rough substrate by a random walk, 
the authors of \cite{cf:FLNO} predicted irrelevance of the disorder, while a few years later \cite{cf:DHV} claimed that the critical point 
was shifted. Both predictions then found supporters in the physics literature until the case was solve mathematically 
(see \cite{cf:GLT10} and references therein).
The full claim in \cite{cf:DHV} is that for the \textsl{wetting} model (see  \eqref{cf:RWpinning}-\eqref{cf:RWwetting} 
for the difference between pinning a wetting)
of a $(p-q)$ random-walk the difference between the quenched critical point $h_c(\gb)$ and the one of the pure system $h_c(0)$ 
satisfies (see \cite[Equation (1.7)]{cf:DHV})
\begin{equation}\label{derridasconjecture}
 \lim_{\gb\to 0+} \gb^2 \log \left( h_c(\gb)-h_c(0)\right)=-\frac{p\pi}{(2-p)^2}.
\end{equation}
The smallness of this conjectured $h_c(\gb)$ for $\gb$ close to zero explains why numerical simulations where not able to produce a general agreement 
between physicists.

\medskip

While the cases covered by the Harris criterion have all been brought on a rigorous ground
\cite{cf:A06, cf:AZ08, cf:DGLT09, cf:GLT, cf:L, cf:T08,cf:T_fractmom}, it turns out that 
the marginal case is still partially open. 
In \cite{cf:GLT10,cf:GLT11} it has been proved that there is indeed a shift in the critical point in the controversial case 
-- $\alpha=1/2$, $\gp(n)$ equivalent to a constant -- 
the best lower bound which is known on $h_c(\gb)$ is 
$\exp\left(-c_b\gb^{-b}\right)$ for all $b>2$ \cite{cf:GLT11}, while the best upper-bound is given by 
$\exp\left(-c\gb^{-2}\right)$ for some non optimal constant~$c$ \cite{cf:A06,cf:T08}.

\medskip

Moreover there remains a very narrow window of slowly varying function for which the issue of disorder relevance is still open
(e.g. $\phi(n)=\sqrt{\log n}$).
The aim of this paper is to settle these two issues by exhibiting a simple necessary and sufficient criterion on the return probability for disorder relevance;
and by proving a generalized version of conjecture \eqref{derridasconjecture}.

\section{Model and results }

\subsection{The disordered pinning model}

We now define in full details the disordered pinning model.
Let $\tau= (\tau_n)_{n\geq 0}$ be a reccurent renewal process, i.e.\ a random sequence  whose increments 
$(\tau_{n+1}-\tau_n)_{n\geq 0}$ are identically distributed positive integers.
We assume that $\tau_0=0$, and that inter-arrival distribution satisfies
\begin{equation}\label{alpha}
\K(n):= \bP(\tau_1 = n) =  (2\pi)^{-1}\gp(n) \, n^{-(1+\ga)}.
\end{equation}
for some $\ga\geq 0$ and slowly varying function $\gp(\cdot)$ 
(the presence of $(2\pi)^{-1}$ in the formula is rather artificial but simplifies further notations).
We denote by $\bP$ the law of $\tau$. With a small abuse of notation we will sometimes consider $\tau$ as a subset of $\bbN$.

\medskip

With no loss of generality, we assume that 
our renewal process is recurrent, i.e.\ that  
$$\bP(\tau_1 = \infty)=1-\sum_{n=1}^{\infty} \K(n)=0.$$
Indeed in the case of transient renewal process, the partition function can be rewritten as of a recurrent renewal, at the cost of a change in the parameter $h$
(see \cite{Chapter 1}[cf:GB]).

\medskip

Let $\go=(\go_n)_{n\in \bbN}$ (the random environment) be a realization of a sequence of IID random variable 
whose law is denoted by $\bbP$.
We assume that the variables $\go_n$ have exponential moments of all order, and set for $\gb\in \bbR$
\begin{equation}\label{defgl}
 \gl(\gb):=\log \bbE[e^{\gb \go}]<\infty.
\end{equation}
We assume (with no loss of generality) that the $\go$'s are centered and have unit variance. 

\medskip

Given $h\in \bbR$ (the pinning parameter), $\gb>0$ (the inverse temperature), and $N\in \bbN$, we define 
a modified renewal measure $\bP^{\gb,h,\go}_{N}$ whose Radon-Nikodym derivative w.r.t.\ $\bP$ is given by
\begin{equation}
\frac{\dd \bP^{\gb,h,\go}_{N}}{\dd \bP}(\tau):=\frac{1}{Z^{\gb,h,\go}_{N}}\exp\left(\sum_{n=1}^N (\gb \go_n+ h-\gl(\gb))\ind_{\{n\in \tau\}}\right) \ind_{\{N\in \tau\}}\, ,
\end{equation}
where $Z^{\gb,h,\go}_{N}$ is the partition function,
\begin{equation}
Z^{\gb,h,\go}_{N}:=\bE\left[\exp\left(\sum_{n=1}^N (\gb \go_n+ h-\gl(\gb))\ind_{\{n\in \tau\}}\right) \ind_{\{N\in \tau\}}\right].
\end{equation}
The free energy per monomer is given by
\begin{equation}
\tf(\gb,h):=\lim_{N\to \infty}\frac 1 N \log Z^{\gb,h,\go}_{N} \stackrel{\bbP-a.s.}{=}\lim_{N\to \infty}\frac 1 N \bbE\left[  \log Z^{\gb,h,\go}_{N}\right].
\end{equation}
See e.g.\ \cite[Theorem 4.1]{cf:G} for a proof of the existence and non-randomness of the limit. 
It is not difficult to check that it is a non-negative, non-decreasing convex function.
When $\gb=0$, there is no dependence in $\go$ and we choose to denote the measure, partition function, and free energy respectively by
$P_N^h$, $Z^h_N$ and $\tf(h)$.
Note that with our convention $\bbE\big[ Z^{\gb,h,\go}_{N}\big]=Z^h_N$, 
so that the partition function and free-energy of the annealed system  (which is obtained by averaging the Boltzmann weight over $\go$) corresponds 
to that of the pure one.

\medskip

The pure free energy has an explicit expression:
\begin{equation}
 \tf(h)=
\begin{cases}
0 \quad \text{ when } h< 0,\\
\tg^{-1}(h) \quad \text{ when } h\ge 0,
\end{cases}
\end{equation}
where $\tg^{-1}$ is the inverse of the function
\begin{equation}
\tg: \  \ \begin{array}{lcl}
 \bbR_+ &\to &\bbR_+,\\
x &\mapsto &-\log \Big( \sum_{n=1}^{\infty} e^{-nx}\K(n)\Big).
 \end{array}
\end{equation}
In particular, this implies that for $\alpha\in (0,1)$
\begin{equation}\label{groomit}
 \tf(h)=h^{\alpha^{-1}}\hat \phi(1/h),
\end{equation}
where $\hat \phi(1/h)$ is an explicit slowly varying function (similar results exists for the cases 
$\alpha\ge 1$ and $\alpha=0$, we refer to \cite[Theorem 2.1]{cf:GB}).
A simple use of Jensen's inequality gives
\begin{equation}
\bbE \left[ \log Z^{\gb,h,\go}_{N}\right] \le \log \bbE\left[ Z^{\gb,h,\go}_{N}\right]=\log Z^h_N , 
\end{equation}
and hence 
\begin{equation}\label{ineq}
 \tf(\gb,h)\le \tf(h).
\end{equation}
Some other convexity property (see \cite[Proposition 5.1]{cf:GB}), on the other hand, implies that  
\begin{equation}
 \tf(\gb,h)\ge \tf(h-\gl(\gb)). 
\end{equation}
Hence the quenched system also presents a phase transition
\begin{equation}
h_c(\gb):=\inf\{ h\in \bbR \ | \ \tf(\gb,h)>0\},
 \end{equation}
and we have 
\begin{equation}
 0\le h_c(\gb)\le \gl(\gb).
\end{equation}
The inequality on the r.h.s.\ is in fact always strict: we have $h_c(\gb)<\gl(\gb)$ (see \cite{cf:AS}).
On the other hand, the question whether $h_c(\gb)$ is equal to zero or not turns out to have a more complex answer, and is deeply related to the 
problem of disorder relevance.

\subsection{Back to the origins: The random walk pinning/wetting models}

Let us also, for the sake of completeness, describe models for pinning/wetting of a simple random-walk which is the one 
introduced and studied in \cite{cf:DHV}.
Given a fixed parameter $p\in(0,1)$, let $\bP$ denote the law of a one dimensional nearest-neighbor simple random walk $S$ on $\bbZ$: 
$S_0=0$ and the increments $X_n:=(S_n-S_{n-1})_{n\ge 1}$ form a sequence of IID variables 
$$\bP(X_n=\pm 1) =p/2 \quad \text{and } \quad \bP[X_n=0]=q=1-p.$$
We define $\bP^{\gb,h,\go}_N$ which is a probability measure defined by its Radon-Nikodym derivative:
\begin{equation}\label{cf:RWpinning}
\frac{\dd \bP^{\gb,h,\go}_{N}}{\dd \bP}(S):=\frac{1}{Z^{\gb,h,\go}_{N}}\exp\left(\sum_{n=1}^{N} (\gb \go_n+ h-\gl(\gb))
\ind_{\{S_n=0\}}\right) \ind_{\{S_{N}=0\}}\, ,
\end{equation}
where $Z^{\gb,h,\go}_{N}$ is the partition function,
\begin{equation}
Z^{\gb,h,\go}_{N}:=\bE\left[\exp\left(\sum_{n=1}^{2N} (\gb \go_n+ h-\gl(\gb))\ind_{\{S_n=0\}}\right) \ind_{\{S_{2N}=0\}}\right].
\end{equation}
We notice that the set 
$\tau:=\left\{ n \ | \  S_{n}=0 \right\}$ is a renewal process. It satisfies for $p\in(0,1)$
(cf. \cite[Proposition A.10]{cf:GB})
\begin{equation}\label{crac}
\bP(\tau_1=n)= \bP[ S_{n}=0\ ; \  S_{k}\ne  0, \forall k\in \{1, \dots, n-1\} ]  \stackrel{n\to\infty}{\sim} \sqrt{\frac{p}{2\pi}} \, n^{-3/2}\, ,
\end{equation}
It thus falls in our framework with $\alpha=1/2$ and $\phi(n)$ converging to $\sqrt{2p\pi}$.
\begin{rem}\
The reader can check that the case $p=1$ which corresponds to the simple random walk on $\bbZ$ 
is a bit different for periodicity issue (the condition $S_N=0$ can only be satisfied for even values of $N$) 
but is equivalent to $p=1/2$ after rescaling space by a factor $2$.
\end{rem}
\medskip

The wetting measure which is the one studied in \cite{cf:DHV} is defined in a similar manner but with the additional constraint that $S$ has to remain positive, to model the presence of a rigid substrate which the interface cannot cross,
\begin{equation}\label{cf:RWpinning}
\frac{\dd \tilde \bP^{\gb,h,\go}_{N}}{\dd  \bP}(S):=\frac{1}{\tilde Z^{\gb,h,\go}_{N}}\exp\left(\sum_{n=1}^{N} (\gb \go_n+ h-\gl(\gb))
\ind_{\{S_n=0\}}\right) \ind_{\{S_{N}=0\ ; \ S_{n\ge 0}, \forall n\in [0,N]\}}\, .
\end{equation}
The constraint has the effect of shifting the pure critical point which is not equal to zero. One has 
\begin{equation}
h_c(0):=\log \frac{2}{2-p}.
\end{equation} 
Even though this is less obvious,
the model also falls in our framework (see \cite[Chapter 1]{cf:GB} for details) and
the associated recurrent renewal process has inter-arrival law.
\begin{equation}\label{croc}
K(n):= \frac{2}{2-p}\bP[ S_{n}=0\ ; \ S_{k}> 0, \forall k\in \{1,\dots, n-1\} ]\stackrel{n\to\infty}{\sim}\frac{1}{2-p}\sqrt{\frac{p}{2\pi}} \, n^{-3/2}\, 
\end{equation}
In particular one has 
$\alpha=1/2$ and $\phi(n)$ converges to $\sqrt{2p\pi}/(2-p)$.

\subsection{Critical point-shift and disorder relevance}

Knowing whether the inequality  $h_c(\gb)\ge 0$ is sharp for small $\beta$ is an important question in terms of disorder relevance.
It corresponds to knowing whether the annealed and quenched critical points coincide.
This question has been the object of a lot of attention of theoretical physicists and mathematicians in the past twenty years
\cite{cf:A06,cf:AZ08,cf:DGLT09,cf:DHV,cf:GLT, cf:GLT10,cf:GLT11, cf:Lmart,cf:T08}.

In \cite{cf:DHV}, Derrida, Hakim and Vannimenus exposed a heuristic argument based on the Harris criterion \cite{cf:Hcrit}, which yields several 
predictions for the critical point shift  for a related hierarchical model. Their claims can 
be translated as follows in the case where the slowly varying function
$\phi$ is asymptotically equivalent to a constant
\begin{itemize}
 \item [(A)] When $\alpha<1/2$ disorder is irrelevant ;
 \item [(B)] When $\alpha>1/2$ disorder is relevant, and $h_c(\gb)$ is of order $\gb^{\frac{2\alpha}{2\alpha-1}}$ ;
 \item [(C)] When $\alpha=1/2$ disorder is relevant, and $-\log h_c(\gb)$ is of order $\gb^{-2}$. 
\end{itemize}
Note that the case $(C)$ presents a special interest, as it includes the pinning of a simple random walk \eqref{cf:RWpinning}. 
Moreover, whereas $(A)$ and $(B)$
have met a general agreement in the physics community, prediction $(C)$ was in opposition to the earlier conclusion of \cite{cf:FLNO},
and remained controversial for a long time (see \cite{cf:GLT10} and references therein).

\medskip

The heuristic argument which is presented in \cite{cf:DHV} is based on second moment computations, and
can easily be generalized for the case of non-trival slowly varying $\phi$ (see e.g the discussion in \cite[Section 1.3]{cf:GLT11}).
The prediction becomes: 
\begin{itemize}
\item [(D)] Disorder is relevant if and only if the renewal process $\tau':=\tau^{(1)}\cap\tau^{(2)}$, 
obtained by intersecting two independent copies of $\tau$, is recurrent.
\end{itemize}
\medskip

Since their publication, these predictions have mostly been brought onto rigorous ground.
In \cite{cf:A06} (see \cite{cf:Lmart,cf:T08} for alternative shorter proofs) it has been shown that when $\tau'$ is terminating 
(i.e.\ is finite), then the disorder is irrelevant.
In \cite{cf:DGLT09} (see also \cite{ cf:AZ10,cf:GLT}), the prediction $(B)$ above was shown to hold true. The existence 
of the limit 
\[c_{\ga} = \lim_{\gb\downarrow 0} h_c(\gb)\gb^{-\frac{2\alpha}{2\alpha-1}}\]
has been proved recently in \cite{cf:CTT}, and it is shown that $c_{\ga}$ is universal, 
in the sense that it does not depend on the law $\bbP$. In \cite{cf:GLT10}, the prediction $(C)$ was partially proved, it was shown that $h_c(\gb)>0$ for all $\beta$ 
with a suboptimal lower-bound. The best standing lower-bound on $h_c(\gb)$ is given in \cite{cf:GLT11}
where is is shown that for any $\gep>0$, $h_c(\gb)\ge e^{-\frac{c_\gep}{\gb^{2+\gep}}}$.

\medskip

Furthermore, the papers \cite{cf:GLT10, cf:GLT11} do not provide a complete proof of prediction $(D)$ but fails very close to it:
for the case $\alpha=1/2$ and $\phi(n)\sim (\log (n))^\kappa$, the method in \cite{cf:GLT11} is sufficient to prove disorder
relevance for $\kappa>1/2$, and is not able to provide give a result only for $\kappa=1/2$ ($\kappa<1/2$ corresponds to disorder irrelevance).

\medskip

In this paper, we prove that $(D)$ holds and prove a sharp estimate for the critical point shift in the case $\alpha=1/2$.

\medskip
We also mention the recent works  \cite{cf:CSZ2, cf:CSZ3} which  proposes an alternative approach to disorder relevance for pinning model.
In the case where  $\tau'=\tau^{(1)}\cap\tau^{(2)}$ is recurrent, the authors consider weak coupling limits of the model by scaling 
$\gb$ and $h$ with $N$ adequately.
 When  $\alpha<1/2$ the right choice is to choose $N$ of the order of the correlation length of the pure system, and $\gb$ such that the variance of the partition function function remains bounded. The existence of a non-trivial scaling limit is derived using the framework of polynomial chaos \cite{cf:CSZ1}.
The case $\ga=1/2$ presents some extra-challenge and is the object of ongoing work \cite{cf:CSZ3}.

\subsection{Results}

Our first theorem is the confirmation of the validity of prediction $(D)$, in the standard interpretation of disorder relevance.

\begin{theorem}
\label{thm:main}
 We have 
 \begin{equation*}
 \left\{ \,  \forall \gb>0, \quad h_c(\gb)>0\,  \right\}
 \end{equation*}
if and only if 
\begin{equation}\label{groupmh}
\sum_{n\ge 1}\frac{1}{n^{2(1-\alpha)}(\phi(n))^2}=\infty.
\end{equation}
\end{theorem}

Note that as mentioned in the previous section, most of the theorem is proved in previous papers, 
and we only need to prove one implication in the case $\alpha=1/2$.

\medskip

Our second result concerns the sharp estimate for the critical point shift in the case 
$\alpha=1/2$ in the absence of slowly varying function. In particular, in view of \eqref{crac} and \eqref{croc} 
it provides a proof of the limit stated in the abstract and of the claim
\eqref{derridasconjecture}.

\begin{theorem}
\label{thm:gap}
Assume that there exists a constant $c_{\gp}$ such that $\lim_{n\to\infty} \gp(n) = c_{\gp}$, or equivalently
$$\K(n)\stackrel{n\to\infty}{\sim} (2\pi)^{-1}\, c_\gp\,  n^{-3/2}.$$
Then we have 
 \begin{equation}\label{daresult}
\lim_{\gb\to 0} \gb^2 \log h_c(\gb) = - \frac12 \, (c_\gp)^2 \, .
 \end{equation}
\end{theorem}

In fact we obtain a more general version of \eqref{daresult} which gives precise asymptotic estimates for arbitrary $\phi$ 
(see Propositions \ref{prop:lowergap}-\ref{prop:uppergap}).

\subsection{Open questions}

Our result completely settles the question of whether a small quantity of disorder induces critical-point shift the pinning model.
However, some issues concerning disorder relevance are still open and even not settled at the heuristic level.
This is in particular the issue of smoothing of the free energy curve.

\medskip

It has been shown in \cite{cf:GT_cmp} that for Gaussian disorder, the growth of the free energy at the vicinity of the critical point is 
at most quadratic
\begin{equation}
 \tf(h,\gb)\le \frac{1+\ga}{2 \gb^2}\, (h-h_c(\gb))^2 .
\end{equation}
When $\alpha<1/2$, in particular, this implies that the free energy curve of the disordered system is \textit{smoother} than that of 
the pure one (this is true in great generality, see \cite{cf:CdH}). For $\alpha>1/2$ instead, it is known \cite{cf:A06, cf:T08, cf:Lmart} that for small $\beta$ the quenched free energy and annealed free energy have the same critical behavior (and in fact more precise results are known \cite{cf:GT09}).
However, it is not known if the quenched and annealed critical behavior coincide in the marginal case $\alpha=1/2$.
Moreover, there is no general agreement on what the critical exponent should be for the disordered system as soon as 
$h_c(\gb)>0$. Let us mention the recent work \cite{cf:DR} where heuristics in favor of an infinitely smooth transition, of the type $\exp\big(-\frac{cst.}{\sqrt{h-h_c(\gb)}}\big)$, are exposed.

\medskip

Note that for our pinning model, critical point-shift of the free energy and smoothing of the free energy curve come together. However,
these two phenomena are not always associated. Let us mention a few variations concerning pinning.
\begin{itemize}
 \item In \cite{cf:QH}, a special case of the pinning model is studied, for which the environment is $\go$ is not IID.
 For this model, there is a smoothing of the free energy curve induced by disorder, but no critical point shift.
 \item In \cite{cf:lighttail}, the case of much lighter renewal $\K(n)\sim e^{-n^{\gamma}}$, $\gamma\in (0,1)$ is considered.
 In that case, there is always a shift of the critical point, but when $\gamma>1/2$, there is no smoothing and the 
 transition of the disordered system is of first order.
 \item In \cite{cf:GL15}, the case of the pinning of the lattice free-field is considered. This is somehow the higher-dimensional
 generalization of the model considered here. It is shown there that for $d\ge 3$, $h_c(\gb)=0$ for all $\beta$,
 but that the quenched free energy grows quadratically at criticality whereas the annealed transition is of first order.
\end{itemize}

\subsection{About the proof}

The proof of Theorem \ref{thm:main} and of the lower-bound in Theorem \ref{thm:gap} are based on the same method, 
which is an improvement of the one used in \cite{cf:GLT11}.
We can divide the process into three steps
\begin{itemize}
\item [(i)] We perform a coarse graining of the system, by dividing it in cells of large size $\ell$.
\item [(ii)] For a power $\theta<1$, we use the inequality $(\sum a_i)^{\theta} \le \sum a_i^{\theta}$ for non-negative $a_i$'s,
in order to reduce the problem to the estimate of the fractional moments of partition functions reduced to a coarse grained trajectory.
\item [(iii)] In the end, we estimate these fractional moments using a change of measure based on a tilt by a multilinear form of the $\go$.
\end{itemize}
Even though we tried to simplify the proof of some technical lemmas, the steps $\rm (i)$ and $\rm (ii)$ are essentially the same than in \cite{cf:GLT11}. 
The novelty lies in the change of measure which is used: instead of using a $q$-linear form with $q$ fixed,
we choose a $q$ which depend on $\ell$ (and we optimize the choice of $q$).

\medskip

This is a rather simple and natural idea, indeed it appears in \cite{cf:GLT11} that a larger value of $q$ gives a better result.
However 
its implementation turns out to be  tricky, as most technical estimates of \cite{cf:GLT11} blow up  much too fast when $q$ goes to infinity. For this reason, 
we have to introduce several refinements which allow us to prove better estimates.
To prove the lower-bound in Theorem \ref{thm:gap}, we have to optimize the constant in several estimates and this has the effect of introducing many small $\gep$'s 
in the computations: this makes the proof of some technical Lemma a bit more delicate. 
Thus, for pedagogical purpuse and readability, it is more suitable to prove first a non-optimal result.

\begin{theorem}
\label{thm:gap2}
Assume that there exists a constant $c_{\gp}$ such that $\lim_{n\to\infty} \gp(n) = c_\gp$.
Then there exists a constant $c_1$ and some $\gb_0>0$ such that, for all $\gb\leq \gb_0$ one has that,
 \begin{equation} \label{marjshift}
h_c(\gb) \geq e^{-c_1 \gb^{-2}}\, .
 \end{equation}
\end{theorem}

\medskip
We now outline the organisation of the rest of the paper, which is divided into two main parts.

\smallskip
 In Sections \ref{sec:coarse}-\ref{sec:chgtmeas}-\ref{sec:oneblock}, we jointly prove  
 Theorem \ref{thm:main} and Theorem~\ref{thm:gap2}.
 Matching asymptotics for $\log h_c(\gb)$ are respectively proved in  Sections \ref{sec:adapt} (lower bound) and \ref{sec:upper} (upper-bound)
 to complete the proof of Theorem \ref{thm:gap}.

\smallskip
- In Section \ref{sec:coarse}, we present the coarse graining scheme: we expose our choice of coarse graining length and, in Proposition \ref{coarsegrained},
 we explain how the proof reduces to having an estimate on the non-integer (\textit{fractional}) moments of partition functions corresponding to 
coarse-grained trajectories.

 - In Section \ref{sec:chgtmeas}, we explain how these fractional moments can be estimated by modifying the law of the environmnent in the blocks  corresponding to 
 the contact points of the coarse grained trajectories. More precisely we show how the proof of our main result reduces to estimating the
 partition function in (a fraction of) one coarse-grained block with modified environment (cf.~Lemma~\ref{lem:oneblock}).
 We also give in Section \ref{sec:descriptionX} a motivated description of the peculiar modification of the environment that we use, which is based on a 
 multilinear form of the $(\go_n)$ with positive coefficients. 
 
 \medskip

- Section \ref{sec:oneblock} is devoted to the more delicate point: the proof of  Lemma \ref{lem:oneblock}. It relies on controlling moments of the multilinear form introduced in the 
previous section. What makes this step difficult is that one has to deal with sums with a very large number of interacting term, which requires several \textit{ad-hoc} tricks to be estimated.

\smallskip

- In Section \ref{sec:upper}, we prove the optimal upper bound on $h_c(\gb)$ of Theorem \ref{thm:gap}, see Proposition \ref{prop:uppergap}. 
The technique is derived from \cite{cf:Lmart} but relies only on a simple second moment computations and does not use Martingale theory (see also \cite[Section 4.2]{cf:G}).

\smallskip

- Finally, in Section \ref{sec:adapt} we adapt the techniques developed in Sections\ref{sec:coarse}-\ref{sec:chgtmeas}-\ref{sec:oneblock}, 
and we obtain the optimal lower bound on $h_c(\gb)$ of Theorem \ref{thm:gap}, see Proposition \ref{prop:lowergap}. 
While it mostly relies on optimizing the constant in the proof, several important modifications are needed: in particular 
we must change the relation between $h$ and the coarse graining length, and prove a substantial improvement of Lemma \ref{lem:oneblock}.

\subsection{On other potential applications of the technique}

This  method combining coarse-graining, fractional moments and change of measure, which originates from \cite{cf:GLT} before being refined in \cite{cf:DGLT09,   cf:GLT10, cf:GLT11, cf:Lhier,cf:T08}, has also been fruitfully adapted for different models:
copolymers \cite{cf:Magic5, cf:BGLT}, random-walk pinning model \cite{cf:QF1, cf:BS1, cf:BS2}, directed polymers and semi-directed polymer 
\cite{cf:L, cf:Nak, cf:Z}, large deviation for random walk in a random environment \cite{cf:YZ},  self-avoiding walk in a random environment \cite{cf:LSAW}.

\smallskip

We believe that the improvement of the method presented in this paper could improve some of the known results in these various areas. 
Let us provide here an example: for the directed polymer model, the adaptation of the method of \cite{cf:GLT11} in \cite{cf:Nak} improved the lower bound 
for the difference between quenched and annealed free energy of the directed polymer  in dimension $1+2$ \cite{cf:L}, from $\exp(-c\gb^{-4})$ to
$\exp(-c_b\gb^{-b})$, $b>2$. It is very likely that the method presented in the present paper, with suitable modification, 
would allow to obtain a bound $\exp(-c\gb^{-2})$ which matches
 the upper bound.

\subsection{Notations}

For $n\in \bbN$, we let $u(n)$ denote the probability that $n$ is a renewal point (by convention we set $u(0)=1$). 
The asymptotic behavior of $u(n)$ was studied by Doney \cite[Thm. B]{cf:Doney},
who proved that for $\alpha\in (0,1)$
\begin{equation}
\label{def:u}
u(n):=\bP(n\in\tau) \stackrel{n\to\infty}{\sim} 2\alpha \sin (\pi \alpha) \gp(n)^{-1} n^{-(1-\ga)},
\end{equation}

When $\ga=\frac12$ (this is the case on which we focus) we obtain 
\begin{equation}
\label{Doney1/2}
 u(n)\stackrel{n\to\infty}{\sim} \frac{1}{\gp(n)\sqrt{n}} \, .
\end{equation}

Recall that $\tau'=\tau^{(1)}\cap \tau^{(2)}$ denotes the intersection of two independent renewals with law $\bP$.
We have 
\begin{equation}
 \bP^{\otimes 2} \big( n\in \tau^{(1)}\cap \tau^{(2)} \big)= u(n)^2
\end{equation}
and thus, from \eqref{def:u}, one deduces that \eqref{groupmh} is equivalent to the reccurence of $\tau'$.
We introduce the quantity
\begin{equation}
D(N):= \sum_{n=1}^{N} u(n)^2.
\end{equation}
Note that $D(N)$ is a non-decreasing sequence.
With some abuse of notation for $x> 0$ 
we set  
\begin{equation}\label{Dinverse}
D^{-1}(x):=\max\{ N \in \bbN \ | \ D(N)\le  x\}.
\end{equation}

When $\lim_{n\to\infty} \gp(n) = c_\gp$, or equivalently $\K(n)\stackrel{n\to\infty}{\sim} (2\pi)^{-1} c_\gp n^{-3/2} $, we have in particular 
\begin{equation}
 D(N) \stackrel{N\to\infty}{\sim} (c_{\gp})^{-2} \log N. 
\end{equation}

\section{Coarse graining and fractional moment}
\label{sec:coarse}

In this section, we explain how our estimate on the critical point shift can be deduced from estimates of 
the fractional moment of partition functions corresponding to coarse grained trajectories.

\subsection{Choice of the coarse-graining length}

We let $\ell$ denote the scale at which our coarse graining will be performed.
Let us fix
$$ A:= 64 e^4\, ,$$
(the choice is quite arbitrary), and set 
\begin{equation}
\label{def:correllength}
\ell_{\gb,A}:=\inf\{ n\in \bbN \ | \ D( \lfloor n^{1/4} \rfloor) \ge A\gb^{-2}\}.
\end{equation}
The reason for this particular choice will appear in the course of the proof.
We are interested in estimating the free-energy for
\begin{equation}\label{defh}
h=h_{\gb,A}:=1/\ell_{\gb,A}. 
\end{equation} 
More precisely our aim is to prove
\begin{proposition}\label{damainpropz}
There exists some $\gb_0>0$ 
such that for all $\gb\leq \gb_0$, 
$$\tf(\gb,h_{\gb,A})=0.$$
\end{proposition}

Note that this is sufficient to prove both Theorem \ref{thm:main}.
and Theorem \ref{thm:gap2}. For Theorem \ref{thm:main}, one just needs to remark that the 
restriction to $\gb\le \gb_0$ does not matter since $\gb\mapsto h_c(\gb)$ is an increasing function 
\cite[Proposition 6.1]{cf:GLT11}.
For Theorem \ref{thm:gap2}, it follows from the definition that 
\begin{equation}\label{soixante4}
 \liminf_{\gb\to 0} \gb^2 \log(h_{\gb,A})\ge  -256\, e^4 (c_{\phi})^2.
\end{equation}

\subsection{Coarse-graining procedure}

The very first step is to transform the problem of estimating the expectation of $\log Z_N$ 
to that of estimating a non-integer moment of $Z_N$.
This is very simply achieved by using the  the concavity of $\log$. We have
\begin{equation}
 \bbE\left[  \log Z^{\gb,h,\go}_N \right]=  \frac{4}{3}\bbE\left[\log (Z^{\gb,h,\go}_N)^{3/4}\right]\le \frac{4}{3} 
 \log \bbE \left[ \left(Z^{\gb,h,\go}_N\right)^{3/4} \right].
\end{equation}
The choice of the exponent $3/4$ here is arbitrary and any value in the interval $(2/3,1)$ would do.
Hence Proposition \eqref{damainpropz} is proved if one can show that
\begin{equation}\label{klfsd}
\liminf_{N\to \infty} \bbE\left[ (Z^{\gb,h_{\gb,A},\go}_N)^{3/4} \right]\le C .
\end{equation}
for some constant $C>0$.

We consider a system whose size $N=m \ell$ is an integer multiple of $\ell$. 
We split the system into blocks of size $\ell$,
\begin{equation}
 \forall i \in \{1,\dots, m \}, \ B_i:=\{\ell (i-1)+1,\ell (i-1)+2, \dots, \ell i\}.
\end{equation}
Given $\cI = \{i_1,\ldots, i_{l}\} \subset \{1,\ldots, m\}$
we define the event 
\begin{equation}
E_{\cI}:= \Big\{ \, \left\{ i\in \{1,\dots, m\} \ | \ \tau\cap B_i\ne \emptyset \, \right\}= \mathcal I  \, \Big\}\, ,
\end{equation}
and set $Z^{\cI}$ to be the contribution to the partition function of the event $E_{\cI}$,
\begin{equation}
Z^{\cI}:= Z^{\gb,h,\go}_{N}(E_{\cI})=Z^{\gb,h,\go}_{N}\bE^{\gb,h,\go}_{N}[E_{\cI}].
\end{equation}
Note that $Z^{\cI}>0$ if and only if $m\in \cI$. When $\tau\in E_{\cI}$, the set $\cI$ is called the coarse-grained trajectory of $\tau$.
As the $E_{\cI}$ are mutually disjoint events, $Z^{\gb,h,\go}_{N}=\sum_{\cI\subset\{1,\dots, m\}}Z^{\cI}$ and thus
using the inequality $(\sum a_i)^{3/4} \leq \sum a_i^{3/4}$ for non-negative $a_i$'s, we obtain
\begin{equation}
\bbE\left[\left( Z^{\gb,h,\go}_{N}\right)^{\frac{3}{4}}\right]\le \sum_{\cI\subset\{1,\dots, m\}} \bbE\left[\left( Z^{\cI}\right)^{\frac{3}{4}}\right].
\end{equation}
We therefore reduced the proof to that of an upper bound on $\bbE\left[\left( Z^{\cI}\right)^{\frac{3}{4}}\right]$, which can be interpreted as the contribution of the coarse grained trajectory $\cI$ to the fractional moment of the partition function.

\begin{proposition}\label{coarsegrained}
Given $\gamma>0$, there exist a constant $\gb_0>0$ such that for all $\gb\le \gb_0$, there exists a constant $C_{\ell}$ which satisfies,
for all $m \ge 1$ and $\cI\subset \{1,\dots,m\}$,
\begin{equation}\label{smallpining}
  \bbE\left[\left( Z^{\cI}\right)^{\frac{3}{4}}\right]\le C_{\ell} \prod_{k=1}^{|\cI|} \frac{\gamma}{(i_{k}-i_{k-1})^{10/9}},
\end{equation}
where by convention we have set $i_0:=0$.
\end{proposition}

\begin{proof}[Proof of Proposition \ref{damainpropz} from Proposition \ref{coarsegrained}]
Now we can notice that the r.h.s of \eqref{smallpining} corresponds (appart from the constant $C_{\ell}$)
to the probability of the a renewal trajectory whose inter-arrival probability is given by
\begin{equation}\label{hatk}
\hat \K(n)={\gamma}{n^{-10/9}},
\end{equation}
provided that the sum is smaller than $1$.
Hence we simply apply the result with
\begin{equation}
 \gamma=\left(\sum_{n=1}^\infty n^{-10/9}\right)^{-1}\, ,
\end{equation}
and we let $\hat\tau$ be the renewal associated to \eqref{hatk}.
With this setup,
\begin{equation}
\sumtwo{\cI\subset \{1,\dots, m \}}{m\in \cI}\bbE\left[\left( Z^{\cI}\right)^{\frac{3}{4}}\right]\le C_{\ell} \bP[m \in \hat \tau]\le C_{\ell}\, ,
 \end{equation}
and thus \eqref{klfsd} is proved.
 \end{proof}

\section{Change of measure}
\label{sec:chgtmeas}

\subsection{Using H\"older's inequality to penalize favorable environments}

The starting idea to prove Proposition \eqref{coarsegrained} is to introduce a change in the law of $\go$ in the cells 
$(B_i)_{i\in \cI}$, which will have the effect of lowering the expectation of $Z^{\cI}$.
Let $g_{\cI}(\go)$ be a positive function of $(\go_n)_{n\in \union_{i\in \cI} B_i}$ 
(that can be interpreted as a probability density if renormalized to have expectation $1$).
Using H\"older's inequality, we have 
\begin{equation}
\label{eq:holder}
\bbE\left[ \left(Z^{\cI}\right)^{3/4}\right] \leq \left(\bbE\left[g_{\cI}(\go)^{-3} \right]\right)^{1/4} \ \left(\bbE\left[g_{\cI}(\go) Z^{\cI} \right]\right)^{3/4}.
\end{equation}
The underlying idea is that most 
of the expectation of $Z^{\cI}$ is carried by atypical environment for which  $Z^{\cI}$ is unusually large.
One should think of applying this inequality to a function $g_{\cI}(\go)$ which is typically equal to one, but which 
takes a small value for the atypical environments which are too favorable.

\medskip

We also want the first term $\bbE[g_{\cI}(\go)^{-3} ]^{1/4}$ to be small, or more precisely, not much larger than one.
Due to our coarse graining procedure, it is natural to choose $g_{\cI}(\go)$ as a product of functions of 
$(\go_n)_{n\in B_i}$, for $i\in \cI$.

\medskip

Our idea, which follows the one introduced in \cite{cf:GLT11}, is to give a penalty when the environment in a block 
has too much ``positive correlation''. The ``amount of correlation'' in a block is expressed using a multilinear form of $(\go_n)_{n\in B_i}$ 
with positive coefficient, 
let us call it $X^i(\go)$ (see \eqref{def:X}). 
We normalize it so that 
\begin{equation}\label{assumpX}
\bbE[(X^i(\go))^2]= 0 \quad \text{and} \quad \bbE[(X^i(\go))^2]\le 1.
\end{equation}
Then we set 
\begin{equation}
\label{def:g}
\begin{split}
 g_i(\go)&:=\exp\left(-M\ind_{\{X^i(\go)\ge e^{M^2}\}} \right),\\
 g_{\cI}(\go)&:=\prod_{i\in B_i}  g_i(\go).
\end{split}\end{equation}
for an adequate value of $M$.

\medskip

With this choice, and for $M\ge 10$, we have for all $i\geq 1$
\begin{equation}
 \bbE\left[g_{i}(\go)^{-3}\right]= 1+(e^{3M}-1)\bbP\left[X^{i}(\go)\ge  e^{M^2}\right]\le 1+(e^{3M}-1)e^{-2M^2}\le 2,
 \end{equation}
and hence
\begin{equation}
\label{costchange}
\bbE[g_{\cI}(\go)^{-3}]= \bbE[g_{1}(\go)^{-3}]^{|\cI|}\le 2^{|\cI|}.
\end{equation}

\subsection{Description of $X_\ell$}
\label{sec:descriptionX}

In order to have a clear idea of the effect of the multiplication by $g_{\cI}$ on the expectation of $Z^\cI$, let us first expand 
$Z^{\cI}$ in order to isolate the contribution of each block. Assume that $\cI:=\{i_1,i_2, \dots, i_l\},$
and let $d_j$ and $f_j$ denote the first and last contact points  in $\tau\cap B_{i_j}$. We have
\begin{equation}
\label{def:hatZ}
Z^{\cI}:= \sumtwo{d_1,f_1 \in B_{i_1}}{d_1\leq f_1}  \cdots \sumtwo{d_l \in B_{i_l}}{f_l=mN}
\K(d_1) u(f_1-d_1) Z^h_{d_1,f_1} \K(d_2-f_1)  \ \cdots\  \K(d_l-f_{l-1}) u(N-d_l ) Z^h_{d_l,N} \, ,
\end{equation}
where we set
\begin{equation}
Z^h_{a,b}:= \bE\left[ \exp\left( \sum_{n=a}^b (\gb \go_n+ h- \gl(\gb)) \, \ind_{\{n\in \tau\}} \right) \ \middle| \ a,b\in\tau \right].
\end{equation}
At the cost of loosing a constant factor per coarse-grained contact point, we can get rid of the influence of $h$:
with our definition of $h=1/\ell$ we have, for all choices of $d_j\le f_j$ in $B_{i_j}$
\begin{equation}
\label{getridofh}
 Z^h_{d_i,f_i}\le e^{h\ell}Z^0_{d_i,f_i}=e Z^0_{d_i,f_i}.
 \end{equation}
Writing $Z_{a,b}$ for $Z^0_{a,b}$ (note that $\bbE[Z_{a,b}]=1$), we have
\begin{equation}
\label{def:hatZ2}
Z^{\cI}\le e^{|\cI|}\sumtwo{d_1,f_1 \in B_{i_1}}{d_1\leq f_1}  \cdots \sum_{d_l \in B_{i_l}}  
\K(d_1) u(f_1-d_1) Z_{d_1,f_1} \K(d_2-f_1)  \ \cdots\  \K(d_l-f_{l-1}) u(N-d_l ) Z_{d_l,N}\, .
\end{equation}
Thus, we want to choose $X^i$ such that for most choices of $d_i,f_i$ one has, for some small $\delta$
\begin{equation}
 \bbE[g_i(\go)Z_{d_i,f_i}]\le \delta  \bbE[Z_{d_i,f_i}]=\delta.
\end{equation}
A natural choice for $X^i$ would be to choose something like
\begin{equation}\label{expression}
 \sum_{d_i,f_i\in B_i} u(f_i-d_i) Z_{d_i,f_i} \, ,
\end{equation}
shifted and scaled to satisfies the assumption on the variance and expectation. Indeed with this choice, $g(\go)$ would be small
whenever the $(Z_{d_i,f_i})$'s are too big.
However, in order to be able to perform the computations, 
it turns out to be better to choose $X$ as a linear combination of products of the $(\go_n)_{n\ge 0}$: hence instead of \eqref{expression}, 
we will choose $X^i$ to be something similar to a high order term in the Taylor expansion in $\go$ of \eqref{expression} (and which is close in spirit to the chaos expansion presented in \cite{cf:CSZ1}).

\medskip

This is the idea behind our choice of $X$:
we perform a formal expansion of
$$u(b-a)\bE \left[\exp\left( \sum_{n=a}^b \left( \gb \go_n- \gl(\gb)\right)\ind_{\{n\in \tau\}} \right) \ \Big|  \ a, b \in \tau \right],$$ 
which is analogous to the Wick expansion of the exponential of a stochastic integral: we drop the normalizing term $\gl(\gb)$ but also all the diagonal terms 
in the expansion.
The term of order $q+1$ is given by 
\begin{equation}
\gb^{q+1} \sum_{a=i_0\le i_1 < i_2 < \dots < i_q=b}  \go_{i_0}u(i_1-i_0)\go_{i_1}\dots u(i_q-i_{q-1})\go_{i_q}.
\end{equation}
Hence after summing along all possibilities for $a$ and $b$ in the block (say $B_1$), and 
renormalizing, we obtain something proportional to 
\begin{equation}\label{grotix}
\sum_{0\le i_0< i_1 < i_2 < \dots < i_q\le \ell}  \go_{i_0}u(i_1-i_0)\go_{i_1}\dots u(i_q-i_{q-1})\go_{i_q}.
\end{equation}
This is the choice of $X$ which was made in  \cite{cf:GLT11}. It was then remarked that when $q$ got larger, it allowed to obtain sharper bounds on $h_c(\gb)$. 
A philosophical reason for this is that by increasing $q$ on gets closer, in some sense, to the expression \eqref{expression}.

\medskip

The novelty of our approach compared to that of \cite{cf:GLT11} is to take $q$ going to infinity with the correlation length $\ell$:
we choose
\begin{equation}
\label{def:q}
q_{\ell}=  \max\Big\{   \log \Big(\sup_{x\leq \ell} \gp(x)\Big) \, ;\, \log D(\ell)  \Big\}
\end{equation} 
Note that for most cases, and in particular when $\gp$ is asymptotically equivalent to a constant, one can take $q=\logg \ell$.
With this choice of growing $q$ the computations become trickier, and to keep them tractable we also choose to reduce the interaction range: we restrict the sum \eqref{grotix} to indices which satisfies $i_j-i_{j-1}\le t_{\ell}$ for some $t_{\ell}\ll \ell$.
We actually take $t=\lfloor \ell^{1/4} \rfloor$, in reference to the definition of the correlation length \eqref{def:correllength}, so that $D(t)\geq A \gb^{-2}$, and $D(t)\leq A \gb^{-2}+1$.
We frequently omit the dependence in $\ell$ in the notation for the sake of readability.

We consider the set of increasing sequences of indices whose increments are not larger than $t$
\begin{equation}
J_{\ell,t}:= \left\{ \u{i}=(i_0,\dots, i_q)\in \bbN^{q+1} \ | \ 1\le i_0<i_1<\cdots < i_q \le \ell \ ; \  
\forall j\in \{1,\dots q\}, i_j-i_{j-1}\le t \right\}. 
\end{equation}
We set
\begin{equation}
\label{def:X}
X^1(\go) := \frac{1}{\ell^{1/2}\, (D(t))^{q/2}}\sum_{\u{i}\in J_{\ell,t}} U(\u{i})\go_{\u{i}}
\end{equation}
where $\go_{\u{i}} = \prod_{k=0}^q \go_{i_k}$, and
\begin{equation}
U(\u{i}) = \prod_{k=1}^q u(i_k-i_{k-1}).
\end{equation}
For $i\ge 1$ we define
\begin{equation}
X^i(\go):=X^1(\theta^{(i-1)\ell}\go))\, ,
\end{equation}
where $\theta$ is the shift operator: $(\theta \go)_n = = \go_{n+1}$.
It is a simple computation to check \eqref{assumpX}, in particular
\begin{equation}\label{eq:varX}
\bbE[(X^i)^2] 
= \frac{1}{\ell\,  (D( t) )^{q}}  \sum_{\u{i} \in J_{\ell,t}} U(\u{i})^2 \leq 1.
\end{equation}

\subsection{Proof of Proposition \ref{coarsegrained}}
\label{sec:endcoarsegraining}

The main step in the proof is to show that this choice of change of measure indeed penalizes the partition function $Z_{d,f}$ for most choices of  $d,f$ inside the cell $B_i$.

\begin{lemma}
\label{lem:oneblock}
For any $M\ge 10$ there exist a choice for the constants $\eta$ and $\gb_0$, such that 
for any $\gb\leq \gb_0$, for all $d,f \in B_i$ with $f-d\ge \eta \ell$, one has
\begin{equation}
\bbE[g_{i}(\go) Z_{d,f}] \leq e^{-M/2} \, .
\end{equation}
\end{lemma}

We prove this result in the next section and explain now how we deduce Proposition \ref{coarsegrained} from it.
Note that, as $g_i(\go)\le 1$, we always have
\begin{equation}
 \bbE[g_{i}(\go) Z_{d_i,f_i}] \le  \bbE[ Z_{d_i,f_i}]= 1.
\end{equation}
Hence, from \eqref{def:hatZ2}

\begin{multline}\label{Z33}
\bbE[g_{\cI}(\go)Z^{\cI}] \\ \le e^{|\cI|}\sumtwo{d_1,f_1 \in B_{i_1}}{d_1\leq f_1} \sumtwo{d_2,f_2 \in B_{i_2}}{d_2\leq f_2} \cdots \sum_{d_l \in B_{i_l}} 
\K(d_1) u(f_1-d_1) e^{-(M/2)\ind_{\{f_1-d_1\ge \eta \ell\}}} \K(d_2-f_1)  \\ \cdots\  \K(d_l-f_{l-1}) u(N-d_l ) e^{-(M/2)\ind_{\{N-d_l\ge \eta \ell\}}}. 
\end{multline}
Using the expression above we are going to show that given $\gamma>0$ and $\gd>0$, if $\ell$ is large enough we have (recall $i_0=0$), for some constant $C_\ell$
\begin{equation}\label{lepremier}
 \bbE[g_{\cI}(\go)Z^{\cI}]\le C_{\ell}(e\gamma^{2})^{|\cI|}\prod_{i=1}^l \frac{1}{(i_j-i_{j-1})^{3/2-2\delta}}.
\end{equation}
Thus, provided that $\delta\le 1/1000 $ and that $\gamma$ is sufficiently small, and together with \eqref{eq:holder} and \eqref{costchange}, it implies that (changing the value of $C_{\ell}$ if necessary)
\begin{equation}
\bbE\left[(Z^{\cI})^{3/4}\right]\le  C_{\ell}\,  \prod_{i=1}^l \frac{\gamma}{(i_j-i_{j-1})^{\frac{10}{9}}}\, .
\end{equation}
We split the proof of \eqref{lepremier} in three parts.
 The main part is to show that for $M$ sufficiently large and $\eta$ sufficiently small,
 for all $j\in \{1,\dots,l-1\}$, for all choices of $f_{j-1}$ ($f_0=0$) and $d_{j+1}$ one has 
 
\begin{multline}\label{gromov}
\sumtwo{d_j,f_j \in B_{i_j}}{d_j\leq f_j}  \sqrt{\K(d_j-f_{j-1})}\, u(f_j-d_j)\, e^{-(M/2)\ind_{\{f_j-d_j\ge \eta \ell\}}}
\sqrt{\K(d_{i+1}-f_{j})}  \\ 
\le \gamma  \left[ (i_j-i_{j-1})(i_{j+1}-i_{j})\right]^{\delta-3/4}
\end{multline}
where $\gamma$ can be made arbitrarily small by choosing $\eta$ small and $M$ large.
We also prove that for all $f_{l-1}$ we have
\begin{equation}\label{crac}
 \sum_{d_l=(i_j-1)\ell+1}^{N}  \sqrt{\K(d_l-f_{l-1})} \, u(N-d_l)\,  e^{-(M/2)\ind_{\{N-d_l\ge \eta\ell\}}} 
    \le \frac{C_{\ell}\gamma}{(m-i_{l-1})^{3/4-\delta}}.
\end{equation}
Finally,  for all $d_1\in B_{i_1}$, we have that
\begin{equation}\label{cric}
 \sqrt{\K(d_1)}\le  \frac{C_\ell}{i_1^{3/4-\delta}}.
\end{equation}
The result \eqref{lepremier} follows by using \eqref{Z33} and multiplying the inequality \eqref{gromov}, \eqref{crac} and \eqref{cric}.
Note that \eqref{cric} is obvious from the property of slowly varying functions.
We focus on the proof of \eqref{gromov} and then show how to modify it to get~\eqref{crac}.

\medskip

\textit{Proof of \eqref{gromov}.}
Let us first consider only the terms of the sum with $f_j$ in the second half of the block: $f_j> (i_j-1/2)\ell$. 
We prove that if $M$ is chosen sufficiently large and $\eta$ sufficiently small, we have for every (fixed) $f_j> (i_j-1/2)\ell$
\begin{equation}\label{karkov}
  \sum_{d_j=(i_j-1)\ell+1}^{f_j}  \sqrt{\K(d_j-f_{j-1})}\,  u(f_j-d_j) \, e^{-(M/2)\ind_{\{f_j-d_j\ge \eta\ell\}}} \\
  \le \frac18\, \gamma\, \frac{\sqrt{\phi((i_j-i_{j-1})\ell)}}{ \ell^{1/4}\, (i_j-i_{j-1})^{3/4} \, \phi(\ell)}. 
\end{equation}
To see this we have to split the sum into two contributions: $d_j\le(i_j-3/4)\ell$ and $d_j\ge(i_j-3/4)\ell+1$.

We observe first that, uniformly on 
$d_j\ge(i_j-3/4)\ell+1$
(so that $d_j- f_{j-1}\geq \frac14 (i_j-i_{j-1})\ell$), we have (provided that $\ell$ is large enough)
\begin{equation}
  \sqrt{\K(d_j-f_{j-1})} \le 4 \frac{\sqrt{\phi((i_j-i_{j-1})\ell)} }{\big( \ell\, (i_j-i_{j-1}) \big)^{3/4}}\, . 
\end{equation}
Hence by summing over $d_j\geq (i_j-3/4)\ell+1$, and using that $u(n) \sim \frac{1}{\sqrt{n} \gp(n) }$ see \eqref{Doney1/2}, we obtain that 
\begin{multline}
\label{minsk0}
 \sum_{d_j=(i_j-3/4)\ell+1}^{f_j}  \sqrt{\K(d_j-f_{j-1})}\,  u(f_j-d_j)\,  e^{-(M/2)\ind_{\{f_j-d_j\ge \eta\ell\}}}\\
 \le  10 \left(\sqrt{\eta}+e^{-M/2}\right)\frac{\sqrt{\ell}}{\phi(\ell)} \, \frac{\sqrt{\phi((i_j-i_{j-1})\ell)} }{\big( \ell\, (i_j-i_{j-1}) \big)^{3/4}}\, . 
\end{multline}

Let us now treat the case of
$d_j\le (i_j-3/4)\ell$, 
in which range $f_j-d_j\ge \ell/4$: one hence has that
 $$ u(f_j-d_j)\ge  3 \, \ell^{-1/2}\phi(\ell)^{-1}.$$ 
Furthermore (one checks separately the cases $i_j-i_{j-1}=1$ and $i_j-i_{j-1}\ge 2$) we have, uniformly in $f_{j-1}$,
\begin{equation}\label{minsk}
\sum_{d_j=(i_j-1)\ell+1}^{(i_j-3/4)\ell}  \sqrt{\K(d_j-f_{j-1})}
  \le   \frac{4 \, \ell^{1/4}\sqrt{\phi((i_j-i_{j-1})\ell)}}{(i_j-i_{j-1})^{3/4}}\, .
\end{equation}
Hence, provided $\eta\le 1/4$ (so that, for $f_j \geq (i_j-1/2)\ell$ and $d_j\leq (i_j-3/4)\ell$ considered, the exponential term is always $e^{-M/2}$), summing over $d_j$ we have
\begin{multline}
\label{minsk2}
\sum_{d_j=(i_j-1)\ell+1}^{(i_j-3/4)\ell}  \sqrt{\K(d_j-f_{j-1})}\, u(f_j-d_j)\, e^{-(M/2)\ind_{\{f_j-d_j\ge \eta\ell\}}}
\\ \le 
 12 \, e^{-M/2}\,  \frac{\sqrt{\phi((i_j-i_{j-1})\ell)}}{\ell^{1/4} \phi(\ell)\, (i_j-i_{j-1})^{3/4}}\, .
\end{multline}
Combining \eqref{minsk0} and \eqref{minsk2} concludes the proof of \eqref{karkov}.
Then, we are ready to sum over $f_j>(i_j-1/2) \ell$: similarly to \eqref{minsk}, we have 
\begin{equation}
 \sum_{f_j = (i_j-1/2)\ell+1}^{i_j\ell} 
 \sqrt{\K(d_{j+1}-f_{j})}\le \frac{4 \ell^{1/4}\,\sqrt{\phi((i_{j+1}-i_{j})\ell)}}{(i_{j+1}-i_{j})^{3/4}}\, .
\end{equation}
Hence, combining this with \eqref{karkov}, we get that
\begin{multline}
  \sumtwo{d_j,f_j \in B_{i_j}}{d_j\le f_j \, f_j >(i_j-1/2)\ell}  
  \sqrt{\K(d_j-f_{j-1})}\, u(f_j-d_j)\, e^{-(M/2)\ind_{\{f_j-d_j\ge \eta\ell\}}} \sqrt{\K(d_{i+1}-f_{j})} \\
  \le \frac12\, \gamma \, \frac{\sqrt{\phi((i_j-i_{j-1})\ell)} }{(i_j-i_{j-1})^{3/4}} \, \frac{1}{\gp(\ell)} \,  \frac{\sqrt{\phi((i_{j+1}-i_{j})\ell)}}{(i_{j+1}-i_{j})^{3/4}}\, . 
\end{multline}

Then we notice that, by symmetry, we can obtain the same bound for the sum over $d_j,f_j \in B_{i_j}$, $d_j\le f_j , d_j \le(i_j-1/2)\ell$,
and thus obtain 
\begin{multline}
  \sumtwo{d_j,f_j \in B_{i_j}}{d_j\le f_j }  
  \sqrt{\K(d_j-f_{j-1})}\, u(f_j-d_j)\, e^{-(M/2)\ind_{\{f_j-d_j\ge \eta\ell\}}} \sqrt{\K(d_{i+1}-f_{j})} \\
  \le \gamma \, \frac{\sqrt{\phi((i_j-i_{j-1})\ell)}}{(i_j-i_{j-1})^{3/4} \sqrt{\gp(\ell)} } \, \frac{\sqrt{\phi((i_{j+1}-i_{j})\ell)}}{(i_{j+1}-i_{j})^{3/4} \sqrt{\gep(\ell)}}\, .
\end{multline}
We deduce \eqref{gromov}  by remarking that, for any $\delta>0$, and if $\ell$ is large enough, one has for all $a\in \bbN$
\begin{equation}
 \frac{\phi(a\ell)}{\phi(\ell)}\le a^{\delta}.
\end{equation}
To prove \eqref{crac}, we notice that the proof of \eqref{karkov} implies that 
\begin{multline}
 \sum_{d_l=(i_j-1)\ell+1}^{N}  \sqrt{\K(d_l-f_{l-1})} \, u(N-d_l)\, e^{-(M/2)\ind_{\{N-d_l\ge \eta\ell\}}} 
  \\ \le \gamma \frac{ \sqrt{\phi((m-i_{l-1})\ell)} }{\ell^{1/4}\phi(\ell)\, (m-i_{l-1})^{3/4}}
  \le \frac{C_{\ell}\, \gamma}{(m-i_{l-1})^{3/4-\delta}}.
\end{multline}
\qed

\section{One block estimate: proof of Lemma \ref{lem:oneblock}}
\label{sec:oneblock}
Due to translation invariance, we may focus on the first block, and for simplicity we write $X(\go)$ (resp. $g(\go)$) instead of $X^1(\go)$ (resp. $g_1(\go)$).
We fix $d,f\in B_1$, with $f-d \geq \eta \ell$, and we set
$$\bP_{d,f}(\cdot) = \bP(\ \cdot  \ | \ d,f\in\tau, \  \tau \cap (B_1\setminus [d,f])=\emptyset ).$$ 
With this notation, one has
\begin{equation}
\bbE[g(\go) Z_{d,f}] = \bE_{d,f}\left[ \bbE\bigg[ g(\go) \, \exp \bigg( \sum_{n=d}^f (\gb\go_n - \gl(\gb) \gd_n)\bigg) \bigg]  \right].
\end{equation}
Note that, given a fixed realization of $\tau$, the exponential in the above quantity averages to one under $\bbP$ and can thus be considered
as a probability density. One introduces the probability measure $\hat{\bbP}_\tau$ 
whose density with respect to $\bbP$ is given by
\begin{equation}
\label{def:Phat}
\hat{\bbP}_\tau (\dd \go) := \exp \bigg( \sum_{n=d}^f (\gb\go_n -\gl(\gb) )\gd_n\bigg)\, \bbP(\dd \go),
\end{equation}
so that
\begin{equation}
\label{eq:newZ}
\bbE[g(\go) Z_{d,f}]  = \bE_{d,f}\left[  \hat{\bbE}_\tau[ g(\go)] \right] .
\end{equation}
Note that, under $\hat\bbP_\tau$, 
$\go$ is still a sequence of independent random variables, but they are no longer identically distributed as the law of $(\go_n)_{n\in \tau\cap \{d,\dots,f\}}$ 
has been exponentially tilted.
This implies in particular a change of mean and variance:
for $d\leq n\leq f$, we have
\begin{equation}
\label{eq:newexpnewvar}
\hat\bbE_\tau [\go_n] =\gl'(\gb) \gd_n,  \quad \Var_{\bbP_\tau} [\go_n] = 1+ (\gl''(\gb)-1)\gd_n 
\end{equation}
where $\gl'$ and $\gl''$ denote the first two derivatives of $\gl$ \eqref{defgl}.
Because of our assumptions on the first two moments of $\go$, one has $\gl'(\gb) \sim \gb$ and $\gl''(\gb) \to 1$ as $\gb\downarrow0$. One hence has, for $\gb$ sufficiently small,
\begin{equation}\label{okbathin}
\gl''(\gb)\in [1/2,2] \quad  \text{ and } \gl'(\gb)/\gb\in [1/2,2].
\end{equation}
For the remainder of the paper, we will always assume that $\gb_0$ is such that \eqref{okbathin} is satisfied for all $\gb<\gb_0$.
We use the notation
\begin{equation}
 \m_{\gb}:=\gl'(\gb) .
\end{equation}
We need to estimate $\bE_{d,f}\left[ \hat{\bbE}_\tau[ g(\go)] \right]$. With the definition \eqref{def:g}, we have 
$$g(\go)\leq \ind_{\{X(\go)\leq e^{M^2}\}} + e^{-M},$$ 
and hence
\begin{equation}
\label{gg}
\bE_{d,f}\left[ \hat{\bbE}_\tau[ g(\go)] \right] \leq \bE_{d,f}\left[ \hat{\bbP}_\tau \big( X(\go)\leq e^{M^2} \big) \right] +e^{-M}.
\end{equation}
One therefore needs to show that under $\bP_{d,f}$, for most realizations of $\tau$, $X(\go)$ is larger than $e^{M^2}$ with $\hat{\bbP}_\tau$-probability close to $1$. 
We obtain this result by estimating the first and second moment of $X(\go)$.
The proof of these two results is quite technical and is postponed to Section \ref{firstmoment} and \ref{secondmomnet} respectively.

\begin{lemma}
\label{lem:newX}
For any $M>0$ and $\eta>0$, there exists some $\gb_0$ such that: for all  $\gb\leq \gb_0$, and for all $d$ and $f$ with $f-d \geq \eta \, \ell_{\gb,A}$,
\begin{equation}
\bP_{d,f}\bigg(  \hat\bbE_\tau[ X ] \geq  2^q \bigg) \geq 1- e^{-M}.
\end{equation}
\end{lemma}

\begin{lemma}
\label{lem:newvarX}
There exists some $\gb_0$ such that, for $\gb\leq \gb_0$ one has
\begin{equation}
\bE_{d,f} \hat\bbE\big[\big( X- \hat\bbE_\tau [X] \big)^2 \big] \leq 3^q.
\end{equation}
\end{lemma}

Note that, due to our definition of $q(A,\gb)$, we have 
\begin{equation}
 \lim_{\gb\to 0+} q(A,\gb)=+\infty\, ,
 \end{equation}
and hence one can always choose $\beta_0$ sufficienty small to have
$$\forall \gb\in (0,\gb_0] \quad e^{M^2} \leq  2^{q-1}.$$

Then, we have
\begin{equation}
\bE_{d,f} \hat\bbP_\tau (X \leq e^{M^2}) \leq \bP_{d,f}\Big(  \hat\bbE_\tau[X] \leq  2^q \Big) 
    + \bE_{d,f} \hat\bbP_\tau \Big( X-\hat\bbE_\tau [ X]  \leq \! - 2^{q-1} \Big).
\end{equation}
We use Lemma \ref{lem:newX} to bound  the first term. The second term can be controlled using  Lemma \ref{lem:newvarX} and Chebychev's inequality. In the end, one obtains
\begin{equation}
\label{eq:afternewX}
 \bE_{d,f} \hat\bbP_\tau (X \leq e^{M^2}) \le e^{-M}+4(3/4)^q.
\end{equation}
Hence from \eqref{gg} we have 
\begin{equation}
 \bE_{d,f}\left[  Z_{d,f}g(\go) \right]\le  2e^{-M}+4(3/4)^q\le e^{-M/2},
\end{equation}
where the last inequality holds if $\gb\leq \gb_0$ with $\gb_0$ chosen sufficiently small.
\qed

\subsection{Proof of Lemma \ref{lem:newX}}\label{firstmoment}

Let us introduce the notation
\begin{equation}
\gd_{\u{i}}:= \prod_{k=1}^q \gd_{i_k} = \prod_{k=1}^q \ind_{\{i_k \in\tau\}}.
\end{equation}
We have
\begin{equation}
\hat\bbE_\tau \left[X \right]= \frac{\m_\gb^{q+1}}{\ell^{1/2}\, (D(t))^{q/2}}\sumtwo{\u{i}\in J_{\ell,t}}{ d\le i_0<i_q \le f} U(\u{i})\gd_{\u{i}}.
\end{equation}
Notice that with our choices for $q$ and $\ell$ we have 
\begin{equation}
\label{usefulbounds}
\gp(\ell)\leq e^{q},   \quad   \sqrt{D(t)}\leq e^q, \quad \text{ and } \gb^2 \geq A/D(t).
\end{equation}
As $\m_{\gb}\ge \gb/2 \geq \tfrac12 \sqrt{A}/\sqrt{D(t)}$, we have
\begin{multline}
\label{simplifyEtauX}
\hat\bbE_\tau [X] \geq \Big(\frac12 \Big)^{(q+1)} A^{(q+1)/2} \frac{1}{\sqrt{D(t)}} \, \frac{1}{\sqrt{\ell} \, D(t)^{q}} \sumtwo{\u{i}\in J_{\ell,t}}{ d\le i_0<i_q \le f} U(\u{i})\gd_{\u{i}}\\
\geq  \left( \frac{\sqrt{A} }{2 e^{2}}\right)^{q+1} \ \times \ \frac{\gp(\ell)}{\sqrt{\ell} \, D(t)^{q}} \sumtwo{\u{i}\in J_{\ell,t}}{ d\le i_0<i_q \le f} U(\u{i})\gd_{\u{i}}\, ,
\end{multline}
Recalling our choice $A=64 e^4$, the r.h.s.\ becomes 
\begin{equation}
4^{q+1} \frac{\gp(\ell)}{\sqrt{\ell} \, D(t)^{q}} \sumtwo{\u{i}\in J_{\ell,t}}{ d\le i_0<i_q \le f} U(\u{i})\gd_{\u{i}}\, .
\end{equation}
Hence to prove Lemma \ref{lem:newX}, it is sufficient to show that
\begin{equation}\label{aaaaaahhhhh}
\bP_{d,f}\Bigg( \frac{\gp(\ell)}{\sqrt{\ell} \, D(t)^{q}} \sumtwo{\u{i}\in J_{\ell,t}}{ d\le i_0<i_q \le f} U(\u{i})\gd_{\u{i}} \ \ \geq\ 2^{-q-2} \Bigg) \geq 1-e^{-M}.
\end{equation}

\medskip

Following \cite[Sec. 5]{cf:GLT11}, we show that we can replace the probability  $\bP_{d,f}$ by $\bP$ and modify slightly the set of indices $J_{\ell,t}$.
This allows to reduce the proving Lemma \ref{lem:W} below. For the sake of completeness we recall the steps.

\begin{itemize}
\item[(a)] By translation invariance, the probability that we have to bound is equal to
\begin{equation}
\bP\Bigg( \frac{\gp(\ell)}{\sqrt{\ell} \, D(t)^{q}} \sumtwo{\u{i}\in J_{\ell,t}}{ 1\le i_0<i_q \le f-d} U(\u{i})\gd_{\u{i}} \ \ \leq\ \ 2^{-q-2} \ \Bigg| \  f-d\in \tau\Bigg).
\end{equation}

\item[(b)] In order to remove the conditioning, we restrict the summation to indices $\u{i}$ such that $i_q\leq (f-d)/2$ and we get an upper bound on the probability. Then 
we use \cite[Lemma A.2]{cf:GLT10} which compares $\bP( \ \cdot \ |  \ n\in \tau)$ to $\bP$, to get that there exists a constant $c_3>0$
such that

\begin{multline}
\label{enlevecontrainte}
\bP\Bigg( \frac{\gp(\ell)}{\sqrt{\ell} \, D(t)^{q}} \sumtwo{\u{i}\in J_{\ell,t}}{ 1\le i_0<i_q \le (f-d)/2} U(\u{i})\gd_{\u{i}} \ \ \leq\ \ 2^{-q-2} \ \Bigg| \  f-d\in \tau\Bigg) \\
\leq c_3 \, \bP\Bigg( \frac{\gp(\ell)}{\sqrt{\ell} \, D(t)^{q}} \sumtwo{\u{i}\in J_{\ell,t}}{ 1\le i_0<i_q \le (f-d)/2} U(\u{i})\gd_{\u{i}} \ \ \leq\ \ 2^{-q-2}\Bigg) .
\end{multline}

\item[(c)] Then, setting $n_{\ell}:=\tfrac14 \eta \ell  \le \frac{f-d}{4}$, we can restrict the summation to indices such that $i_0\leq n$, which automatically ensures that $i_q \leq n+tq \leq (f-d)/2$, provided that $\ell$ is large enough (since $tq\ll (f-d)/4$). Hence, we can replace $J_{\ell,t}$ with 
\begin{equation}
J'_{n,t}:= \left\{ \u{i}=(i_0,\dots, i_q)\in \bbN^{q+1} \ | \ i_0\le n\ ; \  
\forall j\in \{1,\dots q\}, i_j-i_{j-1}\in (0,t]\right\},
\end{equation}
and get an upper bound in the probability \eqref{enlevecontrainte}.

\item[(d)]
Finally, from the definition of $n$ we have,  if $\ell$ is large enough
$$\frac{\gp(\ell)}{\sqrt{\ell} }\ge \frac{\sqrt{\eta}}{4}  \frac{\gp(n)}{\sqrt{n} }\, .$$
Hence \eqref{aaaaaahhhhh} holds provided that we can prove that for all $n\ge \eta \ell$ 
\begin{equation}
\label{aaaaaaaaaaah2}
\bP\Bigg( \frac{\gp(n)}{\sqrt{n} \, D(t)^{q}} \sum_{\u{i}\in J'_{n,t}} U(\u{i})\gd_{\u{i}} \ \ \leq\ \ \eta^{-1/2}2^{-q} \Bigg)\leq \frac{1}{c_3}\, e^{-M}.
\end{equation}
\end{itemize}
Let us set 
\begin{equation}
\label{def:W}
W_{\ell}:=\frac{\phi(n)}{\sqrt{n} \, D(t)^q} \sum_{\u{i}\in J'_{n,t}} U(\u{i})\gd_{\u{i}} \, ,
 \end{equation}
where the dependence in $\ell$ is also hidden in $t$, $q$ and $n$.
We are left to showing the following result which easily yields  \eqref{aaaaaaaaaaah2} for $\gb$ sufficiently small.

\begin{lemma}
\label{lem:W}
Under probability $\bP$, we have 
\begin{equation}
W_{\ell}\ \ \stackrel{ \ell \to\infty}{\Longrightarrow} \quad \frac{1}{\sqrt{2\pi}} |Z|\, , 
 \end{equation}
 where $Z\sim \cN(0,1)$, and $\Longrightarrow$ denotes convergence in distribution.
\end{lemma}

\begin{proof}
First one remarks that the following convergence holds:
\begin{equation}
\label{convsimple}
\frac{\phi(n)}{\sqrt{n}}\sum_{j=1}^{n}  \gd_{j} \ \  \stackrel{n \to\infty}{\Longrightarrow}  \ \ \frac{1}{\sqrt{2\pi}} |Z|\, , \qquad (Z\sim \cN(0,1)) .
\end{equation}
This is a standard result, since $$\left\{\sum_{j=1}^{n}  \gd_{j} <m \right\} = \{\tau_m >n\}\, ,$$ 
so that 
\[\lim_{n\to\infty} \bP\left( \frac{\phi(n)}{\sqrt{n}}\sum_{j=1}^{n}  \gd_{j} > t\right) = \bP(\sigma_{1/2} >1/t^2) \, ,\]
where $\sigma_{1/2}$ is an $1/2$-stable subordinator at time $1$.
(such a remark was already made in \cite{cf:GLT11}, see Equations (5.12)-(5.16)).

\medskip

The lemma is thus proved if one can show that the difference 
\begin{equation}
\label{def:DW}
  \gD W_{\ell}:=\frac{\phi(n)}{\sqrt{n}} \left( \sum_{j=1}^{n}  \gd_{j}-\frac{1}{D(t)^q}\sum_{\u{i}\in J'_{n,t}} U(\u{i})\gd_{\u{i}}\right),
\end{equation}
converges to zero in probability, thanks to Slutsky's Theorem. 
We simply prove that the second moment of  $\gD W_{\ell}$ tends to zero.
Set
\[ J'_{n,t}(j):= \left\{ \u{i} \in  J'_{n,t} \ | \ i_0=j \right\} \]
and
\begin{equation}
\label{def:Y}
Y_j:=\gd_{j}-\frac{1}{D(t)^q}\sum_{\u{i}\in J'_{n,t}(j)} U(\u{i})\gd_{\u{i}},
\end{equation}
so that $\gD W = \frac{\phi(n)}{\sqrt{n}}  \sum_{j=0}^n Y_j$.

\begin{lemma}
\label{lem:Y}
 We have the following estimates:
 
 \begin{itemize}
  \item [(i)] for $|j_1-j_2| > tq$, \quad
$ \bE[Y_{j_1}Y_{j_2}]=0$ ;

 \item [(ii)] there exists a constant $C_1>0$ such that for all $j\geq 0$, \quad
$ \bE [Y^2_{j}]\le (C_1)^q \bE[\delta_j]  =  (C_1)^q\, u(j). $
 \end{itemize}
\end{lemma}
 Using this result we have
\begin{multline}
\label{eq:afterY}
 \bE [\gD W^2]=\frac{\phi(n)^2}{n}\sum_{j_1,j_2=0}^{n}   \bE[Y_{j_1}Y_{j_2}]
\leq  \frac{2\phi(n)^2}{n}\sum_{j_1=0}^{n} \sum_{j_2=j_1}^{j_1+tq}  \bE[Y_{j_1}Y_{j_2}] \\
\leq  \frac{2\phi(n)^2}{n} (C_1)^q \sum_{j_1=0}^{n} \sum_{j_2=j_1}^{j_1+tq}  u(j_1)^{8/9} u(j_2)^{1/9} ,
 \end{multline}
where in the first inequality, we used (i), and in the second one we used H\"older's inequality, together with (ii). Since there exists a constant $c_4$ such that $u(j)\leq c_4 (1+j)^{-9/20}$ for all $j\geq 0$,  we have that, provided that $\ell$ is large enough,
\begin{equation}
 \bE [\gD W^2] \leq  \frac{2 c_4\phi(n)^2}{n} (C_1)^q  ( tq +1) \sum_{j_1=0}^{n} (1+j_1)^{- 2/5} \leq 10 c_4 \phi(n)^2 (C_1)^q t q n^{-2/5}.
\end{equation}
Note that with our choice of $q$, 
$$(C_1)^q = \left( \max \big\{ \sup_{x\leq \ell} \gp(x), D(\ell) \big\} \right)^{\log C_1}$$ 
is a slowly varying function of $\ell$. Since $t=\lfloor \ell^{1/4} \rfloor$ and $n\geq \tfrac14 \eta \ell$ we obtain 
\begin{equation}
\bE[\gD W^2] \leq \big(10 c_4 \phi(n)^2 q \, (C_1)^q \eta^{-2/5} \big)  \times \ell^{-3/20}\, ,
\end{equation}
which goes to $0$ as $n\to\infty$.
\end{proof}

 \subsection{Proof of Lemma \ref{lem:Y}}\label{secondmomnet}
 We introduce a new notation. 
 If $\u{i}$ and $\u{j}$ are finite increasing sequences of finite cardinal $q+1$ and $q'+1$ we let $\u{i}\u{j}$ denote the increasing sequence whose image is given by the union of that of $\u{i}$ and $\u{j}$. Note that the cardinal of $\u{i}\u{j}$ is not necessarily equal to $q+q'+2$, as it is possible that $i_{k}$ and $j_{k'}$ coincide.
 
 \medskip
 
 We also extend the definition $U(\u{i})$ to increasing sequences  $(i_k)^{0\leq k \leq r}$ of arbitrary (finite) cardinal (recall that $u(0)=1$ by convention)
 
 \begin{equation}
 U(\u{i}):=\prod_{k=0}^{r} u(i_{k}-i_{k-1})\, .
 \end{equation}

 For item $(i)$, we write
 \begin{equation}
   \bE[Y_{j_1}Y_{j_2}]= 
   \bE\left[ Y_{j_1} \delta_{j_2} 
   \left( 1-\frac{1}{D(t)^q}\sum_{\u{i}\in J'_{n,t}(j_2)} U(\u{i})\gd_{\u{i}}\right) \right].
 \end{equation}
Conditioned to $\delta_{j_2}=1$, and assuming  that $j_2> j_1+tq$, one has that $Y_{j_1}$ and
\[ \Big( 1-\frac{1}{D(t)^q}\sum_{\u{i}\in J'_{n,t}(j_2)} U(\u{i})\gd_{\u{i}}\Big)\]
are independent. The latter term have mean zero (condintionally on $\gd_{j_2}=1$), hence the conclusion.

\medskip

For item $(ii)$ conditioning to $\delta_j=1$ and using translation invariance, one obtains

\begin{equation}
  \bE [Y^2_{j}]=\bE[\delta_j]
  \left(\frac{1}{D(t)^{2q}}\sum_{\u{i}, \u{k} \in J'_{n,t}(0)} U(\u{i}) U(\u{k})
  \bE\left[\gd_{\u{i}}\, \gd_{\u{k}} \right]-1\right)
\end{equation}

In order to keep track of the role of $q$ in the definition of $J'_{n,t}(0) (\subset \bbN^{q+1})$, we now write $J'_{n,t,q}$ instead.
We prove that there exists a constant $C_1$ such that, for any couple $q, q'$
 \begin{equation}
 \label{eq:sumUUU}
 \sumtwo{\u{i}  \in J'_{n,t,q}(0)}{\u{k}  \in J'_{n,t,q'}(0)}\!\!  U(\u{i}) U(\u{k})U(\u{i}\, \u{k})\le (C_1)^{q+q'} D(t)^{q+q'}.
 \end{equation}
 
 This is obviously true if $q=q'=0$, and we proceed recursively on $q+q'$.
 We decompose the sum into two components according to whether $i_q$ or $k_{q'}$ is larger.
 In the case $i_q\ge k_{q'}$
 one obtains
 \begin{equation}
 \sumtwo{\u{i}  \in J'_{\ell,t,q-1}(0)}{\u{k}  \in J'_{n,t,q'}(0)}\!\! U(\u{i}) U(\u{k})U(\u{i}\, \u{k}) \sum_{i_q=\max(i_{q-1},k_{q'})}^{i_{q-1}+t} u(i_q-i_{q-1})u\big( i_q-\max(i_{q-1},k_{q'}) \big) \, .
 \end{equation}
Now, note that thanks to \eqref{def:u} there exists a constant $c_5$ such that, for all $m\geq n$, one has $u(m)\leq c_5 u(n)$. Therefore, uniformly in the choice of $\u{i}$ and $\u{k}$, we have that 
 \begin{multline}
  \sum_{i_q=\max(i_{q-1},k_{q'})}^{i_{q-1}+t} u(i_q-i_{q-1})u\big(i_q-\max(i_{q-1},k_{q'})\big) \\
  \le  c_5 \sum_{i_q=\max(i_{q-1},k_{q'})}^{i_{q-1}+t} u \big(i_q-\max(i_{q-1},k_{q'})\big)^2 \leq 
  c_5 D(t).
 \end{multline}
By symmetry, we conclude that
\begin{equation}
\label{eq:lastsumUUU}
\sumtwo{\u{i}  \in J'_{n,t,q}(0)}{\u{k}  \in J'_{n,t,q'}(0)} \!\! U(\u{i}) U(\u{k})U(\u{i}\, \u{k})\le  2c_5 
\max\left\{  \sumtwo{\u{i}  \in J'_{n,t,q-1}(0)}{\u{k}  \in J'_{n,t,q'}(0)}\!\!  U(\u{i}) U(\u{k})U(\u{i}\, \u{k}) \ ; \!\!\!\!  \sumtwo{\u{i}  \in J'_{n,t,q}(0)}{\u{k}  \in J'_{n,t,q'-1}(0)} \!\! U(\u{i}) U(\u{k})U(\u{i}\, \u{k})\right\}\, ,
\end{equation}
which in turns gives \eqref{eq:sumUUU} by induction, with $C_1=2c_5$.

\subsection{Proof of Lemma \ref{lem:newvarX}}
\label{sec:newvarX}

We set $$\hat\go_i = \go_{i} -\m_\gb \gd_i \ind_{\{d\leq i\leq f\}}.$$
 Under $\hat\bbP_\tau$, the $\hat\go_i$'s are independent, centered random variables, with $\bbE[\hat\go_i^2] \leq 2$ (recall  \eqref{eq:newexpnewvar}). 
 We have
\begin{equation}\label{kouskous}
\Var_{\hat\bbP_{\tau}}\big[ X\big]
 = \frac{1}{\ell\, (D(t))^{q}} \hat\bbE_{\tau} \Bigg[
 \bigg( \sum_{\u{i}\in J_{\ell,t}} U(\u{i})\prod_{k=0}^q  (\hat\go_{i_k} + \m_\gb \gd_{i_k})\bigg)^2 \Bigg]
 -  \frac{\m_\gb^{2(q+1)} }{\ell\, (D(t))^{q}} \Bigg( \sum_{\u{i}\in J_{\ell,t}}  U(\u{i}) \gd_{\u{i}}  \Bigg)^2 
 \end{equation}
 
 One can develop the product, for some fixed $\u{i}\in J_{\ell,t}, d\leq i_0 < i_q\leq f$
 \begin{equation}
  \prod_{k=0}^q  (\hat\go_{i_k} + \m_\gb \gd_{i_k})=   \sum_{r=0}^{q+1} \m^r_\gb \sumtwo {A\subset \{0,\dots, q\}}{ |A|=r}
  \bigg(\prod_{j\in A} \gd_{i_j}\bigg)  \bigg(\prod_{k\in  \{0,\dots, q\}\setminus A} \hat\go_{i_k}\bigg),
 \end{equation}
 so that, when developing the square, and taking the expectation we have
\begin{multline}\label{bigsomme}
\hat\bbE_{\tau} \Bigg[ \bigg( \sum_{\u{i}\in J_{\ell,t}}  U(\u{i})\prod_{k=0}^q  (\hat\go_{i_k} + \m_\gb \gd_{i_k})\bigg)^2  \Bigg] \\
= \sum_{\u{i} ,\u{i'}\in J_{\ell,t}}    U(\u{i}) U(\u{i}') 
\sum_{r}^{q+1} \m^{2r}_\gb \sumtwo{A,B\subset \{0,\dots, q\}}{ |A|=|B|=r}
  \bigg(\prodtwo{j\in A}{j'\in B} \gd_{i_j} \gd_{i'_{j'}}\bigg)  \hat\bbE_{\tau} \left[ \prodtwo{k\in  \{0,\dots, q\}\setminus A}{k'\in  \{0,\dots, q\}\setminus B} \hat\go_{i_k} \hat \go_{i'_{k'}}\right].
\end{multline} 
We have used the fact that only $|A|$ and $|B|$ with the same cardinality have non-zero expectation.
Note that the sum of the terms with $r=q+1$ corresponds exactly to $\bbE\big[ X\big]^2$ and thus just cancels the second term in the r.h.s of \eqref{kouskous}.
 
 \medskip

 To get a good bound on the expected value of \eqref{bigsomme} we must reorganize it. In the process we will also add some positive term, 
 but this is not a problem since we work on an upper bound.
In  \eqref{bigsomme}, because of the last factor, the non-zero terms have to satisfy
$$(i_k)_{k\in \{0,\dots ,q\} \setminus A}=(i_k)_{k\in \{0,\dots ,q\} \setminus B}\, .$$ 
 For a given $s$, we define the set $M_{s}$, which includes all the 
  values that can be taken by $(i_k)_{k\in  \{0,\dots, q\}\setminus A}$, when $q+1-|A|=s$.
  
  \medskip

We notice that
 $$  \bE_{d,f} \bigg(\prodtwo{j\in A}{j'\in B} \gd_{i_j} \gd_{i'_{j'}}\bigg) =0$$
 if one of the $i_j$-s or $i'_{j'}$-s is out of the interval $[d,f]$.
 Hence if $r\ge 0$, the non-zero terms must also satisfy 
 \begin{equation}
 i_0\ge d-tq, \quad i_q\le f+tq.
 \end{equation}
 We include this condition in the definition of $M_s$ 
 \begin{equation}
  M_{s}:= \big\{ \u{i}\in \bbN^{s} \ | \   d-tq \le i_0<\dots<i_{s}\le  f+tq \,  , \, \forall k, l\in\{1,\dots, s\}, | i_k-i_{l} |\le t q \big\}. 
 \end{equation}
 Note that this definition will result in adding extra terms in the sum  if either $d< tq$  or  $f>\ell -tq$  (as we dropped the condition $i_0\le 0$).
 Then, given $\u{i}$, we define $N_{r}(\u{i})$  which includes all the 
  values that can be taken by $(i_j)_{j\in A}$ with $|A|=r$.
 We say that an increasing sequence of integer  $\u{m}=(m_0,\dots, m_a)$ is \textit{$t$-spaced} if 
 \begin{equation}
 \forall k\in \{ 1,\dots, a\}, m_k-m_{k-1} \in (0,t] \, .
 \end{equation}
We set 
  \begin{equation}
  N_r(\u{i} ):= \{ \u{j}\in \bbN^{r} \ | \  d \le j_1<\dots<j_{r}\le f,   \quad \u{i} \cap \u{j}=\emptyset, \quad  \u{i} \u{j} \text{ is $t$-spaced}  \}. 
 \end{equation}
 where  $\u{i} \cap \u{j}=\emptyset$ means that the images of the sequences $\u{i}$ and  $\u{j}$ are disjoint.

\medskip

 With this notation, and using  \eqref{okbathin} or more specifically
 $$\hat\bbE_\tau[\hat\go_{i_k}^2] \leq 2 \quad \text{ and } \quad \m_{\gb}\leq 2\gb\, ,$$ 
 we have  
 \begin{equation}
 \label{VarXdeveloped}
 \Var_{\hat \bbP_{\tau}}\big( X\big)
 \leq  \frac{2^{q+1}}{\ell\, (D(t))^{q}}  \sum_{\u{i}\in J_{\ell,t}} U(\u{i})^2 
 + \frac{4^{q+1}}{\ell\, (D(t))^{q}} \sum_{r=1}^{q} \gb^{2r} 
\sum_{\u{i}\in M_{q-r}}  \ \sum_{\u{j},\u{k}\in N_r(\u{i})} U(\u{i}\u{j})U(\u{i}\,\u{k}) \delta_{\u{j}}\delta_{\u{k}}
\end{equation}
where we isolated the term with $r=0$, and we have used  the concatenation notation
$\u{i}\u{j}$ introduced in the proof of Lemma \ref{lem:Y}.
The first term in the r.h.s.\ is equal to (recall \eqref{eq:varX})
\begin{equation}
2^{q+1} \bbE\left[X^2 \right]\le 2^{q+1}.
\end{equation}
We end the proof by controlling the contibution to the sum of the others terms, which turns out to be ridiculously small in expectation.

\begin{lemma}\label{lestermesenr}
There exists constants $C_2$ and $c_6$ such that for all $r\in \{1,\dots, q-1\}$
\begin{equation}
\label{lasttobebounded}
\sum_{\u{i}\in M_{q-r}}\sum_{\u{j},\u{k}\in N_r(\u{i})} U(\u{i}\u{j})U(\u{i}\,\u{k}) \bE_{d,f}\big[\delta_{\u{j}}\delta_{\u{k}}\big]
\le  c_6\, q\, (C_2 D(t))^{q+r-1} \frac{\sqrt{t\ell}}{\phi(t)\phi(\ell)}.
\end{equation}
\end{lemma}

We have thus 
\begin{multline}
\label{boundvariance}
  \frac{4^{q+1}\gb^{2r}}{\ell\, (D(t))^{q}}
  \sum_{\u{i}\in M_{q-r}}\sum_{\u{j}\, ,\, \u{k}\in N_r(\u{i})} U(\u{i}\u{j})U(\u{i}\,\u{k}) \bE_{d,f}\big[\delta_{\u{j}}\delta_{\u{k}}\big]\\
  \le c_6\, q 4^{q+1} (C_2 )^{2q}( D(t))^{r-1} \gb^{2r} \frac{\sqrt{t}}{ \sqrt{\ell} \phi(t)\phi(\ell)}\, .
\end{multline}
Now by the definition \eqref{def:correllength} of $\ell$ we have  $\gb^2 D(t) \leq 2 A$. Using also that $t \leq \ell^{1/4}$, the above sum is smaller than 
\begin{equation}
 c_6 \, q \gb^2  ( 8 C_2^2 A )^{q+1}\frac {\ell^{-3/8} }{\phi(t)\phi(\ell)}\le  \ell^{-1/4}.
\end{equation}
The last inequality is valid provided that $\ell$ is large enough, since $q$, $(8 C_2^2 A)^q$, $\gp(t)$ and $\gp(\ell)$ are slowly varying functions.
Hence, from \eqref{VarXdeveloped}, we have 

\begin{equation}
 \bE_{d,f}\left( \Var_{\hat \bbP_{\tau}}\big( X\big) \right)=  2^{q+1}+ q  \ell^{-1/4}\, ,
\end{equation}
which concludes the proof of Lemma \ref{lem:newX}, provided that $\ell$ is large enough. \qed

\begin{proof}[Proof of Lemma \ref{lestermesenr}]
Remark that if $i_0\le (d+f)/2$, then $j_r,k_r \leq (d+f)/2 +tq$ so that $f- \max(j_r,k_r) \geq (f-d)/4$ (provided $\ell$ is large enough). Since there exists  a constant $c_7>0$ such that $u(m)\leq c_7 u(n)$ whenever $m\geq \tfrac14 n$ (recall \eqref{def:u}), one obtains
\[\bE_{d,f}\left[\delta_{\u{j}}\delta_{\u{k}}\right] =  U(d \u{j} \u{k}) \, \frac{u(f- \max(j_r,k_r)) }{u(f-d)} \leq c_7 U(d \u{j}\u{k}) \, .\]

By symmetry, the contribution to the sum of  \eqref{lasttobebounded}  of $\u{i}\in M_{q-r}$ such that $i_0\le (d+f)/2$ is equal to that of 
$\u{i}\in M_{q-r}$ such that $i_{q-r}\ge (d+f)/2$, and hence the whole sum is bounded above by

\begin{multline}
\label{starttotrim}
2c_7 \sumtwo{\u{i}\in M_{q-r} }{d-tq \leq i_0\leq (d+f)/2}
\sum_{\u{j}\, ,\, \u{k}\in N_r(\u{i})}  U(\u{i}\u{j})U(\u{i}\, \u{k})U(d\u{j} \u{k})\\
 \leq 2c_7 c_5  \sum_{i_0=d-tq}^{d+f/2} u(\max(i_0-tr-d,0))\sum_{\u{i}\in M_{q-r}(a)}
\sum_{\u{j}\, ,\, \u{k}\in N_r(\u{i})}   U(\u{i}\u{j})U(\u{i}\, \u{k})U(\u{j}\u{k}) ,
\end{multline}
where we used that $\min\{j_1,k_1\} \geq \max\{ i_0 -tr,d\}$, so that $u(\min\{j_1,k_1\} -d) \leq c_5 u(\max\{ i_0 -tr,d\} -d) $.
We also used the notation 
$$M_s(a) = \{\u{i}\in M_s  \, ;\, i_0=a\}.$$
Then, one has that there exists a constant $c_8$ such that for any $\ell$
\begin{equation}
\label{eq:c8}
\sum_{i_0=d-tq}^{d+f/2} u(\max(i_0-rt-d,0))\le \sum_{n=0}^{\ell} u(n)  \le c_8 \frac{\sqrt{\ell}}{\phi(\ell)},
\end{equation}
which follows from the fact that  $(d-f)\le \ell$ and classical properties of regularly varying functions.
Combining this with \eqref{starttotrim}, \eqref{eq:c8} and  Lemma \ref{lem:ontrime} below proves Lemma \ref{lestermesenr}, with the constant $c_6 = 4 c_7 c_5 (c_8)^2$ and $C_2= 3c_5$.
\end{proof}

\begin{lemma}
\label{lem:ontrime}
For any $a\in \bbZ$, and any $s\geq 1$, $r_1,r_2\geq 1$ 
 \begin{equation}
 \sum_{\u{i}\in M_{s}(a)} \sum_{\u{j}\in N_{r_1}(\u{i})} \sum_{\u{k}\in N_{r_2}(\u{i})}     U(\u{i}\u{j})U(\u{i}\, \u{k})U(\u{j}\u{k}) \le (1+s)\,  (3 c_5 D(t))^{s+r_1+r_2-1} \times 2c_8
 \frac{\sqrt{t}}{ \phi(t)}  \, .
\end{equation}
We recall that the constant $c_5$ is chosen such that for all couples of integers such that $m\ge n$, we have  $u(m)\leq c_5 u(n)$, and the constant $c_8$ appears in \eqref{eq:c8}.
 \end{lemma}

 \begin{rem}\label{precision}\rm
The result is proved by induction and we also have to consider the case  where either $r_1$, $r_2$ or $s$ is equal to zero.
When $r_1$ or $r_2$ are equal to zero, the definition of $N_r$ is extended as follows: 
$N_{0}(\u{i})=\{\emptyset\}$ if $\u{i}$ is $t$-spaced and $N_{0}(\u{i})=\emptyset$ if not. We will also use the convention 
$U(\emptyset)=1$.
 \end{rem}

\begin{proof}[Proof of Lemma \ref{lem:ontrime}]
Note that there is in fact no dependence in $a$, and one can as well set $a=0$.
We now proceed with a triple induction on the indices $s$, $r_1$ and $r_2$.
Let us start with the induction hypothesis.
We set 
\begin{equation}
 \Sigma(s,r_1,r_2):=  \sum_{\u{i}\in M_{s}(0)} \sum_{\u{j}\in N_{r_1}(\u{i})} \sum_{\u{k}\in N_{r_2}(\u{i})}     U(\u{i}\u{j})U(\u{i}\, \u{k})U(\u{j}\u{k})
\end{equation}

{\bf (1)} We first show that if $r_1,r_2, s\ge 1$, we have
\begin{equation}
\label{recursilon}
  \Sigma(s,r_1,r_2)\le   c_5 D(t)\Big[\Sigma(s,r_1,r_2-1)+  \Sigma(s,r_1-1,r_2)+  \Sigma(s-1,r_1,r_2)\Big]
\end{equation}

To see this we decompose the sum  $\Sigma(s,r_1,r_2)$ into three sums $\Sigma_{k}$, $\Sigma_j$ and $\Sigma_i$ 
corresponding to the respective contributions  of the triplets $\u{i}$, $\u{j}$, $\u{k}$ satisfying
$k_{r_2}\ge  \max(i_s,j_{r_1})$,  $j_{r_1}\ge  \max(i_s,k_{r_2})$, and $i_{s}\ge  \max(j_{r_1},k_{r_2})$ respectively.
As we are counting several times the cases of equality between $j_{r_1}$ and $k_{r_2}$,
 we have
 
   $$\Sigma(s,r_1,r_2)\le  \Sigma_{k}(s,r_1,r_2)+\Sigma_j (s,r_1,r_2)+ \Sigma_i(s,r_1,r_2).$$

To bound   $\Sigma_{k}$ from above, we notice that because of the restriction of the sum to the  $\u{i}\, \u{k}$ which are $t$-spaced, we have
\begin{multline}
 \Sigma_{k}(s,r_1,r_2) =\sum_{\u{i}\in M_{s}(0)} \sum_{\u{j}\in N_{r_1}(\u{i})} \sum_{\u{j}\in N_{r_2-1}(\u{i})} \\
 \sum_{k_{r_2}=
 \max(i_{s}+1,j_{r_1},k_{r_2-1}+1)}^{\max(i_s, k_{r_2-1})+t}u\left(k_{r_2}-\max(i_s,k_{r_2-1})\right)
u\left(k_{r_2}-\max(j_{r_1},k_{r_2-1})\right).
\end{multline}
Then for any value of $i_s$, $k_{r_2-1}$ and $j_{r_1}$ we have
\begin{multline}\label{cacpadur}
\sum_{k_{r_2}= \max(i_s+1,j_{r_1},k_{r_2-1}+1)}^{\max(i_s, k_{r_2-1})+t}u\left(k_{r_2}-\max(i_s,k_{r_2-1})\right)
u\left(k_{r_2}-\max(j_{r_1},k_{r_2-1})\right)\\
\le c_5  \sum_{k_{r_2}= \max(i_s+1,j_{r_1},k_{r_2-1}+1)}^{k_{r_2-1}+t}
u\left(k_{r_2}-\max(i_s+1,j_{r_1},k_{r_2-1}+1)\right)^2 \le c_5 D(t).
\end{multline}

In the case where $r_2=1$, we just have to drop $k_{r_2-1}$ from the $\max$ and sum until 
$i_s+t$. We therefore have that
\begin{equation}
  \Sigma_{k}(s,r_1,r_2)\le c_5 D(t) \Sigma_{j}(s,r_1,r_2)\le c_5 D(t)\Sigma(s,r_1-1,r_2)\, ,
\end{equation}
and the exact same proof yields 
\begin{equation}
\begin{split}
 \Sigma_{j}(s,r_1,r_2)\le c_5 D(t)\Sigma(s,r_1-1,r_2)\, ,\\
  \Sigma_{i}(s,r_1,r_2)\le c_5 D(t)\Sigma(s-1,r_1,r_2)\, .
\end{split}
\end{equation}

{\bf (2)} Now let us treat the case where either $r_1$, $r_2$ or $s$ are equal to $0$.
The technique of spliting the sum according to the type of the largest index as above still works and gives
\begin{equation}
\label{recursilon2}
\begin{split}
 \Sigma(s,r_1,0)&\le  c_5 D(t) \Big[\Sigma(s,r_1-1,0)+  \Sigma(s-1,r_1,0)\Big],\\
  \Sigma(0,r_1,r_2)&\le  c_5 D(t) \Big[\Sigma(0,r_1,r_2-1)+  \Sigma(0,r_1-1,r_2)\Big],
\end{split}
\end{equation}
provided that $s\ge 1$, $r_1\ge 2$ in the first case, and $r_2,r_1\ge 1$ in the second case (note that from Remark \ref{precision} we sum only over $t$-spaced $\u{i}$).

{\bf (3)} To finish the induction we are left with proving bounds on $\Sigma(0,r_1,0)$, and $\Sigma(s,1,0)$.

{\bf a.} For the first one, when $r_1\ge 2$ ,
we split the sum into two contribution  $j_{r_1} > 0$ or $j_{r_1}\le 0$.
They are respectively equal to
\begin{equation*}\
\begin{split}
\Sigma_{>0}(0,r_1,0)&=\sumtwo{\u{j}\in N_{r_1-1}(0)}{j_{r_1}\geq 0} U(\u{j})U(0\u{j}) \sum_{j_{r_1}=\max(0,j_{r_1-1})+1}^{\max(j_{r_1-1},0)+t}u(j_{r_1}-j_{r_1-1})u(j_{r_1}-\max(j_{r_1-1},0)),\\
\Sigma_{\le 0}(0,r_1,0)&=\sumtwo{ \u{j}\in N_{r_1-1}(0) }{ j_{r_1}\le 0 } U(\u{j})U(0\u{j}) \sum_{j_{0}=j_{1}-t}^{j_1} u(j_{1}-j_{0})^2.
\end{split}
\end{equation*}
And similarly to \eqref{cacpadur}, it is sufficient to conclude that 
\begin{equation}
 \Sigma(0,r_1,0)\le (c_5+1) D(t)\Sigma(0,r_1-1,0) \leq 2c_5\, D(t)\,  \Sigma(0,r_1-1,0) \ .
\end{equation}
Moreover, one also has that
\begin{equation}
\label{eq:3}
 \Sigma(0,1,0)\le 2\sum_{j_{1}=1}^{t} u(|j_{1}|)\le 2 c_8 \frac{\sqrt{t}}{ \phi(t)}\, .
\end{equation}
Then, one easily has by induction that for any $r_1\geq 1$ (a similar result holds for $r_2\geq 1$), 
\begin{equation}
\label{eq:4}
\Sigma(0,r_1,0) \leq \big( 2 c_5 D(t) \big)^{r_1-1} \times 2 c_8 \frac{\sqrt{t}}{ \phi(t)}  \leq  \big( 3 c_5 D(t) \big)^{r_1-1} \times 2 c_8 \frac{\sqrt{t}}{ \phi(t)} \, .
\end{equation}

{\bf b.}
It is straightforward to check that 
\begin{equation}
\label{eq:5}
 \Sigma(s,0,0)= \sumtwo{\u{i}\in M_s(0)}{\u{i} \ t-\text{spaced}} U(\u{i})^2= D(t)^s.
\end{equation}
Moreover, for all $s\ge 1$, decomposing the sum according to whether $j_1$ or $i_s$ is larger, we have
\begin{multline}
\label{eq:6}
  \Sigma(s,1,0)\le \sumtwo{\u{i}\in M_s(0)}{\u{i} \ t-\text{spaced}} U(\u{i})^2\sum_{j_1=i_s+1}^{i_s+t} u(j_1-i_s) 
  \\
  + \sumtwo{\u{i}\in M_{s-1}(0)}{\u{i} \ t-\text{spaced}}\sum_{\u{j}\in N_1{\u{i}}} U(\u{i})U(\u{i}\u{j})  \sum_{i_s=\max(i_{s-1},j_s)+1}^{\max(i_{s-1},j_1)+1} u(i_s-i_{s-1})
  u(i_s-\max(i_{s-1},j_1))
  \\ \le
  c_8 \frac{\sqrt{t}}{ \phi(t)}\, \Sigma(s,0,0) + c_5 D(t)\, \Sigma(s-1,1,0).
\end{multline}
Therefore, combining \eqref{eq:3}-\eqref{eq:5}-\eqref{eq:6}, one easily gets by induction that, for any $s\geq 0$,
\begin{equation}
\label{eq:7}
  \Sigma(s,1,0) \leq  (1+s)\, (c_5 D(t))^{s} \times 2c_8 \frac{\sqrt{t}}{ \phi(t)}  \leq (1+s)\, \big( 3 c_5 D(t) \big)^{s} \times 2 c_8 \frac{\sqrt{t}}{ \phi(t)} \, .
\end{equation}

{\bf (4)} One is now able to complete the induction by combining \eqref{recursilon}-\eqref{recursilon2} with \eqref{eq:4}-\eqref{eq:7}.

\end{proof}

\section{Upper bound of Theorem \ref{thm:gap}}
\label{sec:upper}

In this Section, we prove the following.
\begin{proposition}
\label{prop:uppergap}
For every $\gep>0$, there exists some $\gb_\gep$ such that, for all $\gb\leq \gb_{\gep}$, one has
\[h_c(\gb) \leq  D^{-1}\big( (1-\gep)/\gb^2 \big)^{-\frac12 (1-\gep)}  \, .\]
In the case where $\lim_{n\to\infty} \gp(n) = c_{\gp}$, we have 
\begin{equation}
\limsup_{\gb\to 0} \gb^2 \log h_c(\gb)\leq- \frac12 (c_\gp)^2  \, .
\end{equation}
\end{proposition}

The proof we present here relies on ideas developped in \cite{cf:Lmart} but we got rid of the use of martingale result, to focus only on simple 
 second moment estimates. We  optimize it here in order to obtain the exact order for $\log h_c(\gb)$ when $\ga=1/2$.

First of all, one establishes a finite volume criterion for localization, see \eqref{eq:finitevolume}.
Then, one proves that the measure $\bP_{N}^{\gb,h=0,\go}$ is close enough to $\bP$ (in some specific sense, see Lemma \ref{lem:measure}), provided that the second moment of the partition function at $h=0$ is not too large. Then Lemma \ref{lem:variance} provides an estimate on $\bbE[(Z_{N}^{\gb,0,\go})^2]$, which, combined with the finite volume criterion, leads to an upper bound on the critical point.

\bigskip
In this section, for technical convenience, we work with the \emph{free} boundary condition. We introduce the measure
$\bP_{N,\free}^{\gb,h,\go}$, and its associated partition function $Z_{N,\free}^{\gb,h,\go} $, which corresponds to this boundary condition 
(in which the constraint $\ind_{\{N\in\tau\}}$ is dropped):
\begin{equation}
\frac{\dd \bP^{\gb,h,\go}_{N,\free}}{\dd \bP}(\tau):=\frac{1}{Z^{\gb,h,\go}_{N,\free}}\exp\left(\sum_{n=1}^N (\gb \go_n+ h-\gl(\gb))\gd_n\right) \, ,
\end{equation}
with 
\begin{equation}
 Z_{N,\free}^{\gb,h,\go} : = \bE \Big[ \exp \Big( \sum_{n=1}^N (\gb\go_n -\gl(\gb) +h )\gd_n\Big) \Big].
\end{equation}

\subsection{Finite-volume criterion for localization}

We notice that we can obtain a bound on the free-energy which is directly related to the contact fraction at the critical point

\begin{lemma}
\label{lem:finitevol}
For all $N$ sufficiently large, for all $h\ge 0$ and all $\gb\in[0,1]$ we have 
\begin{equation}
\tf(\gb,h) \ge \frac h N \bbE\bE^{\gb,0,\go}_{N,\free}\left[\sum_{i=1}^N \delta_n \right]- \frac{2\log N}{N}. 
\end{equation}
As a consequence, for all $N$ sufficiently large, for all $\gb\in [0,1]$ we have
\begin{equation}\label{hcc}
 h_c(\gb)\le \frac{2\log N}{\bbE\bE^{\gb,0,\go}_{N,\free}\left[\sum_{i=1}^N \delta_n \right]}.
\end{equation}
\end{lemma}

\begin{proof}
It is the result of a simple computation (see \cite[Ch. 4]{cf:GB}) that there exists a constant $c_{10}$ such that, for all $h\ge 0$
\begin{equation}
Z_{N}^{\gb,h,\go} \leq Z_{N,\free}^{\gb,h,\go} \leq c_{10} N e^{\gb |\go_N|} \, Z_N^{\gb.h,\go},
\end{equation}

Then, by super-additivity of the expected log-partition function, we have
\begin{equation}
\label{eq:finitevolume}
\tf(\gb,h)  = \sup_{N\in\bbN} \frac1N \bbE \log Z_{N}^{\gb,h,\go} \geq \frac1N \bbE \log Z_{N,\free}^{\gb,h,\go} - \frac{\log(c_{10} N) +\gb}{N}
\end{equation}
the last inequality being valid for any $N\geq 1$. 
Finally, by convexity  we note that for any $h>0$
\begin{equation}\label{finad}
 \log Z_{N,\free}^{\gb,h,\go}\ge  h \, \partial_u \log Z_{N,\free}^{\gb,u,\go} |_{u=0}+   \log Z_{N,\free}^{\gb,0,\go}.
\end{equation}
The last term is larger than 
\begin{equation}
 \log \bP (\tau_1 > N) \ge -\log (N/ c_{10})+1,
\end{equation}
provided $N$ is large enoug.  Also, a basic computation yields
\begin{equation}
 \partial_u \log Z_{N,\free}^{\gb,u,\go} |_{u=0}= \bE^{\gb,0,\go}_{N,\free}\left[\sum_{i=1}^N \delta_n \right].
 \end{equation}
Hence we get the result by combining \eqref{eq:finitevolume} and \eqref{finad}.
\end{proof}

\subsection{Estimating the contact fraction at criticality}

Now, to estimate $\bbE\bE^{\gb,0,\go}_{N,\free}\left[\sum_{i=1}^N \delta_n \right]$, we need to compare it with the pure system.
The underlying idea is the following: for the pure system (for $h=0$ it is just the law $\bP$), the number of contact is of order $N^{1/2}\phi(N)^{-1}$.
We want to show that, as long as the second moment of the partition function $Z_{N,\free}^{\gb,0,\go}$ is not too big, the order of magnitude for the number of contacts remains the same for the disordered system.

\begin{lemma}
\label{lem:measure}
For all $\gep>0$, there exists some $N_\gep$ such that, if $N\geq N_\gep$ and $\bbE\big[ (Z_{N,\free}^{\gb,0,\go})^2 \big] \leq 10/\gep$, then
\begin{equation}
\bbE\bigg[  \bP_{N,\free}^{\gb,0,\go} \Big(  \sum_{n=1}^N \gd_n \geq N^{\frac{2-\gep}{4}}\Big) \bigg] \geq  \frac{\gep}{80}\, .
\end{equation}
\end{lemma}
\begin{proof}
We denote $A_N=\big\{ \sum_{n=1}^N \gd_n \geq N^{\frac{2-\gep}{4}}\big\}$. from \eqref{convsimple}
we have
\begin{equation}
\lim_{N\to \infty} \bP(A_N)=1\, ,
\end{equation}
and hence we can find $N_\gep$ such that for all $N\geq N_\gep$ 
$$\bP(A_N)\geq 1- \gep/160.$$
Then, we observe that
\begin{equation}
\bP_{N,\free}^{\gb,0,\go}(A_N^c) \leq \ind_{\{Z_{N,\free}^{\gb,0,\go} \leq 1/2\}} + 2 \bE\Big[ \ind_{A_N^c} \, e^{\sum_{n=1}^N (\gb\go_n - \gl(\gb)) \gd_n}\Big] \, .
\end{equation}
Therefore, averaging over the disorder and using Paley-Zygmund's inequality for $Z_{N,\free}^{\gb,0,\go}$ (recall $\bbE[Z_{N,\free}^{\gb,0,\go}]=1$), we have that
\begin{equation}
\bbE\big[\bP_{N,\free}^{\gb,0,\go}(A_N^c) \big] \leq \bbP( Z_{N,\free}^{\gb,0,\go} \leq 1/2 ) + 2 \bP(A_N^c) \leq 1 -  \frac{1}{4 \bbE\big[ (Z_{N,\free}^{\gb,0,\go})^2\big]} + \frac{\gep}{80} \, .
\end{equation}

Hence, if $\bbE\big[ (Z_{N,\free}^{\gb,0,\go})^2\big] \leq  10/\gep$, one concludes that  
\begin{equation}
\bbE\big[\bP_{N,\free}^{\gb,0,\go}(A_N^c) \big] \leq 1- \gep/80.
 \end{equation}
\end{proof}

Given $\gep>0$, we set
\begin{equation}
\label{def:Nbeta}
N_{\gb,\gep} := \max \Big\{\, N  \, ; \,  \bbE\big[ (Z_{N,\free}^{\gb,h=0,\go})^2 \big] \leq 10/\gep\, \Big\} \, .
\end{equation}
If $\gb$ is chosen sufficiently small, we can ensure that $N_{\gb,\gep}\ge N_{\gep}$ of Lemma \ref{lem:measure}.
Hence we have 
\begin{equation}
 \bE^{\gb,0,\go}_{N_{\gb,\gep},\free}\left[\sum_{i=1}^{N_{\gb,\gep}} \delta_n \right]\ge \frac{\gep}{80}\,  N_{\gb,\gep}^{\frac{2-\gep}{4}}.
\end{equation}

And recalling Lemma \ref{lem:finitevol}, and in particular \eqref{hcc}, one has
\begin{equation}
h_c(\gb) \leq \frac{160}{\gep}\, (\log N_{\gb,\gep})N_{\gb,\gep}^{-\frac{2-\gep}{4}}\le N_{\gb,\gep}^{-\frac{1-\gep}{2}} \, ,
\end{equation}
where the last inequality holds provided $N_{\gb,\gep}$ is sufficiently large.

To conclude the proof of Proposition \ref{prop:uppergap}, we need a control of $N_{\gb,\gep}$.
\begin{lemma}
\label{lem:variance}
For every $\gep>0$, there exists $\gb_{\gep}$ such that for all $\gb\in (0,\gb_\gep]$ 
\begin{equation}
N_{\gb,\gep} \geq D^{-1}\big( (1-\gep)/\gb^2 \big) \geq N_{\gep}\, .
\end{equation}

\end{lemma}

\subsection{Control of the second moment: proof of Lemma \ref{lem:variance}}

One needs to control the growth of $\bbE\big[ (Z_{N,\free}^{\gb,0,\go})^2\big]$: we show that if $N$ is such that $D(N)\leq (1-\gep)/\gb^2$, then $\bbE\big[ (Z_{N,\free}^{\gb,0,\go})^2\big]\leq 10/\gep$.

First of all, one writes 
\begin{equation}
(Z_{N,\free}^{\gb,0,\go})^2 = \bE^{\otimes 2}\Big[ \exp\Big( \sum_{n=1}^N (\gb\go_n - \gl(\gb) ) (\gd_n^{(1)} + \gd_n^{(2)})\Big) \Big],
\end{equation}
where $\tau^{(1)}$ and $\tau^{(2)}$ are two independent copies of $\tau$, whose joint law is denoted by $\bP^{\otimes 2}$ and $\delta^{(i)}=\ind_{\{n\in \tau^{(i)}\}}$.
Therefore, since
\begin{equation}
 \log\bbE[e^{(\gb\go_n - \gl(\gb) )p} ]=\begin{cases}
                                       0 &\text{ for } p= 0 \text { or } 1,\\
                                       \gl(2\gb)-2\gl(\gb) &\text{ for } p=2,
\end{cases}\end{equation}
we have
\begin{equation}
 \bbE\left[ Z_{N,\free}^{\gb,0,\go})^2\right] = 
 \bE^{\otimes 2} \Big[\exp\Big( \sum_{n=1}^N (\gl(2\gb) - 2\gl(\gb) ) (\gd_n^{(1)}\gd_n^{(2)})\Big) \Big].
\end{equation}

As
$$(\gl(2\gb) - 2\gl(\gb)) \stackrel{\gb\to 0}{\sim} \gb^2 \, ,$$ 
there exists some $\gb_{\gep}$ such that, if $\gb\leq \gb_{\gep}$, then
$$(\gl(2\gb) - 2\gl(\gb))\leq (1+\gep^2) \gb^2.$$
Hence we have 
\begin{equation}
\bbE\big[ (Z_{N,\free}^{\gb,0,\go})^2\big] \leq \bZ_{N}^{(1+\gep^2)\gb^2},
\end{equation}
where $\bZ_{N}^{u}$ is the partition function (with free-boundary condition)
of a homogeneous pinning model with parameter $u$  and underlying renewal $\tau'$, obtained by intersecting two independent copies of $\tau$,
$\tau' := \tau^{(1)} \cap\tau^{(2)}$ : 
\begin{equation}
 \bZ_{N}^u := \bE^{\otimes 2} \big[ e^{u \sum_{n=1}^N \ind_{\{n\in\tau'\}}} \big].
\end{equation}
We rewrite  $\bZ_{N}^u$ in the following manner 
\begin{equation}
\bZ_{N}^u=1+ \sum_{k=1}^N (e^{ku} - e^{(k-1)u}) \bP^{\otimes 2} (|\tau'\cap[0,N]| \geq k).
\end{equation}

Then, to obtain an upper bound we use the following trivial fact 
\begin{equation}
\bP^{\otimes 2} \big( |\tau'\cap[0,N]| \geq k \big) \leq \big( \bP^{\otimes 2} (\tau'_1 \le N) \big)^k\, .
\end{equation}
Hence we have
\begin{equation}\label{dsadsa}
 \bZ_{N}^u\le 1 + u \sum_{k=1}^N \exp \Big( k \left[u +\log \bP^{\otimes 2} (\tau'_1 \le N) \right] \Big).
\end{equation}
To estimate the tail of the distribution of $\tau'_1$ we use Theorem 8.7.3 in \cite{cf:BGT} (recall that $\bP^{\otimes 2}(n\in\tau') =u(n)^2$): we have that, since $D(N)$ is slowly varying,
\begin{equation}
\sum_{n=1}^N \bP^{\otimes 2}(n\in\tau') = D(N) \quad \Longrightarrow \quad \bP^{\otimes 2}(\tau'_1 \geq n) \stackrel{N\to\infty}{\sim} \frac{1}{D(N)} \, .
\end{equation}

In the end, we obtain that, provided that $N$ is large enough,
\begin{equation}
  \log \bP^{\otimes 2} (\tau'_1 \le N)\le -\frac{\left(1- \frac{\gep}{4}\right)}{ D(N)},
\end{equation}
so that, from \eqref{dsadsa}, we get
\begin{equation}
\bZ_{N}^u \leq 1+ u\sum_{k=1}^N e^{\frac{k}{D(N)} (  u D(N)- (1-\gep/4))}.
\end{equation}
Now recall that we wish to use the inequality for $u= \gb^2(1+\gep^2)$.
If $D(N) \leq (1-\gep)/\gb^2$, then $D(N)\leq (1-\gep/2)/u$ provided that $\gep$ is small,
and we have $$ u D(N)- (1-\gep/4)\leq - \gep/4. $$ 
In the end, we obtain
\begin{equation}
\bZ_N^u \leq 1+ \, \frac{u}{1- \exp\left(-\frac{\gep}{4D(N)} \right)} \leq 1+ \frac{8}{\gep} u D(N) \leq 10/\gep,
\end{equation}
where in the second inequality, we used that $\gep/(4D(N))$ is small if $N$ is large enough. 
The last inequality is valid  if $D(N) \leq (1-\gep/2)/u$, and $\gep$ is small enough.
\qed

\section{Optimizing the lower bound for Theorem \ref{thm:gap}}

\label{sec:adapt}
In this Section, we sharpen the argument of Sections \ref{sec:coarse}-\ref{sec:chgtmeas}-\ref{sec:oneblock}, and prove the following (recall \eqref{Dinverse}).

\begin{proposition}
\label{prop:lowergap}
For all $\gep>0$, there exists $\gb_\gep>0$ such that, for all $\gb\leq \gb_{\gep}$, one has 
\[h_c(\gb) \geq  D^{-1}\big( (1+\gep)/\gb^2 \big)^{-\frac12 (1+\gep)} \, . \]
As a consequence when $\lim_{n\to\infty} \gp(n) = c_{\gp}$, we have 
\begin{equation}\label{grocrocro}
 \liminf_{\gb\to 0}\gb^2\log h_c(\gb)\ge - \frac12 ( c_{\gp})^2.
\end{equation}

\end{proposition}

Comparing \eqref{grocrocro} to \eqref{soixante4}, we realize that we have to gain a factor $512\, e^4$ in the limit.
One can gain a factor two by choosing $h_{\ell}$ to be such that $\ell$ is much closer to the annealed correlation length
$1/\tf(h)$ which is, up to slowly varying correction asymptotically equivalent to $h^{-2}$ (cf. \ \eqref{groomit}).
We choose to have  
$$h := \ell^{-(1+\gep/2)\frac12}.$$

\smallskip

A factor $64 e^4$ is gained by choosing $A=1+\gep$ instead of $64 e^4$.
Also, instead of taking $t=\ell^{1/4}$, we choose $t$ to be a power of $\ell$ as close as $1$ as necessary: we take $t=\ell^{1-\gep^2}$, which yields a extra gain of a factor $4$.

This implies to introduce some modifications to optimize the proof, which we summarize below.

\begin{itemize}
\item [(i)] in the definition of $q$, if we choose to multiply by a large factor, say $\gep^{-2}$, we can avoid to losing 
the exponential factors in \eqref{usefulbounds}. The net benefit of the operation is a factor $e^4$ in the choice of $A$.

\item [(ii)] in \eqref{okbathin}, we can replace $2$ and $1/2$ by quantities which are arbitrarily close to $1$ provided that $\beta$ is chosen small enough (a factor $4$ is gained). More precisely we fix $\gb_{\gep}$ such that 
for all $\gb\in (0,\gb^{\gep})$
\begin{equation}\label{okbathin2}
 \gl''(\gb)\in \big[e^{-\gep^2},1+\gep^{3}/2 \big] \text { and }  \frac{1}{\gb} \gl'(\gb)\in \big[ e^{-\gep^2},1+\gep^{3}/2 \big].
\end{equation}

\item [(iii)] in  Lemmas \ref{lem:newvarX}, and \ref{lem:newX} we do not need $2^q$  and $3^q$ and they can be replaced by powers arbitrarily close to one.
\item[(iv)] in \eqref{aaaaaahhhhh}, we do not need $2^{-q-2}$, one can replace it with any quantity which tends to zero (factor $4$ again).
\end{itemize}
The change which has the more serious consequence is the modification of $t$: we need some refinements to prove Lemmas \ref{lem:W} and 
\ref{lem:Y}. After this brief sketch we now present the modifications in details.

\medskip

We set
\begin{equation}
\label{def:correllength2}
\ell_{\gb,\gep}:=\inf\big\{ n\in \bbN \ | \ D( \lfloor n^{1-\gep^2} \rfloor) \ge (1+\gep) /\gb^{2} \big\}.
\end{equation}

Let us take $h_{\gb,\gep}:= \ell_{\gb,\gep}^{-\frac{2+\gep}{4}}$, 
and prove that there exists some $\gb_\gep$ such that
\begin{equation}
\forall \gb\in (0, \gb_{\gep}], \quad \tf(\gb,h_{\gb,\gep})=0.
\end{equation}
This is enough to obtain Proposition \ref{prop:lowergap} provided that $\gep$ is small enough to satisfy 
$$(2+\gep) \leq 2(1+\gep)(1-\gep^2).$$

\subsection{Adaptation of the change of measure}

The coarse-graining and fractional moment arguments are identical, and no modification is needed there.
The change of measure argument works also in the same manner: the choice of $g_{\cI}$ is the same as in \eqref{def:g}, and the functional $X(\go)$ is also the same as in \eqref{def:X}, 
except for our choice of $t$ and $q$,
\begin{equation}
\label{def:tq2}
\begin{split}
t_{\ell}  &:= \lfloor \ell^{1-\gep^2} \rfloor  \, ; \\
q_{\ell} &:= \frac{1}{\gep^2} \max\Big\{   \log \Big(\sup_{x\leq \ell} \gp(x)\Big) \, ;\, \log D(\ell)  \Big\}\, .
\end{split}
\end{equation}
As \eqref{getridofh} is obviously not valid for our choice of $h$,
we have now to prove a variant
of Lemma \ref{lem:oneblock}, with a partition function which still includes the parameter $h$. 
\begin{lemma}
\label{lem:oneblock2}
For any $M\ge 10$ there exists some $\eta$ and some $\gb_\gep$, such that 
for any $\gb\leq \gb_\gep$,  $d,f \in B_i$, we have
\begin{equation}
\bbE[g_{i}(\go) Z_{d,f}^h] \leq
\begin{cases}e^{-M/2} \text{ if } f-d\ge \eta \ell,          \\                      
  2 \text{ if } f-d\le \eta \ell \, .                     
  \end{cases} 
\end{equation}
\end{lemma}

This results allows to show that, similarly to \eqref{Z33},
\begin{multline}
\bbE[g_{\cI}(\go)Z^{\cI}] \le 2^{|\cI|}\sumtwo{d_1,f_1 \in B_{i_1}}{d_1\leq f_1} \cdots \sum_{d_l \in B_{i_l}} 
\K(d_1) u(f_1-d_1) e^{-(M/2)\ind_{\{f_1-d_1\ge \eta \ell\}}} \K(d_2-f_1)  \\ \cdots\  \K(d_l-f_{l-1}) u(N-d_l ) e^{-(M/2)\ind_{\{N-d_l\ge \eta \ell\}}},
\end{multline}
and we can then follow the proof of Section \ref{sec:chgtmeas} to conclude.

\medskip

The core of the proof of Lemma \ref{lem:oneblock2} is the use 
of adapted versions of Lemmas \ref{lem:newX} and \ref{lem:newvarX}.
\begin{lemma}
\label{lem:newX2}
With the updated choice of $X(\go)$ (with $\ell$ as in \eqref{def:correllength2}, $t$ and $q$ as in \eqref{def:tq2}), one has that, for any $M\geq 11$ and $\eta>0$, there exists some $\gb_\gep$ such that, for all $\gb\leq \gb_\gep$, and all $d\leq f$ with $f-d\geq \eta \ell$,
\[\bP_{d,f} \Big( \hat\bbE_\tau [X(\go)] \geq (1+\gep^2)^q\Big) \geq 1-e^{-M} .\] 
\end{lemma}

\begin{lemma}
\label{lem:newvarX2}
With the updated choice of $X(\go)$ (with $\ell$ as in \eqref{def:correllength2}, $t$ and $q$ as in \eqref{def:tq2}),
there exists some $\gb_\gep$ such that, for $\gb\leq \gb_\gep$ one has
\begin{equation}
\bE_{d,f} \hat\bbE\big[\big( X- \hat\bbE_\tau [X] \big)^2 \big] \leq (1+\gep^3)^{q}.
\end{equation}
\end{lemma}

\begin{proof}[Proof of Lemma \ref{lem:oneblock2}]

Let us start with the second case $f-d\le \eta \ell$.
Note that for any choice of $d,f$, we have (as $g_i(\go)\le 1$)
\begin{equation}
\bbE[g_{i}(\go) Z_{d_i,f_i}^h] \leq \bE_{d,f}\left[e^{h \sum_{i=0}^{d-f} \gd_i} \right].
\end{equation}
Up to a factor $e^h$ this corresponds to the partition function of the homogeneous pinning model.
Now, as we have chosen $h$ such that $f-d$ is much smaller than the correlation length, we can use the bound from \cite[Equation (A.12)]{cf:DGLT09}
to obtain that for $\gb$ small enough, for all $f-d\le \ell_{\gb,\gep}$.
\begin{equation}\label{troctoc}
\bE_{d,f}\left[e^{3h \sum_{i=d}^{f} \gd_i} \delta_{f-d} \right]\le 2.
\end{equation}
For the case $f-d\ge \eta \ell$ with the same definition of $\hat\bbP_{\tau}$ as in \eqref{def:Phat}, we have that, similarly to \eqref{eq:newZ},
\begin{equation}
\bbE[g_{i}(\go) Z_{d,f}^h] = \bE_{d,f} \bigg[ e^{h\sum_{n=d}^f \gd_n} \hat\bbE_\tau [g(\go)]  \bigg].
\end{equation}
Hence, using the definition \eqref{def:g} of $g(\go)$, one gets that
\begin{multline}
\label{eq:oneblock2}
\bbE[g_{i}(\go) Z_{d,f}^h] \leq  \bE_{d,f} \Big[ e^{h\sum_{n=d}^f \gd_n}  \hat\bbP_\tau\big(X(\go) \leq e^{M^2} \big) \Big] + \bE_{d,f} \big[ e^{h\sum_{n=d}^f \gd_n}  \big] e^{-M}\\
\leq \bE_{d,f} \big[ e^{3h\sum_{n=d}^f \gd_n} \big]^{1/3} \bE_{d,f} \big[  \hat\bbP_\tau\big(X(\go) \leq e^{M^2} \big)^{3/2} \big]^{2/3} + 2 e^{-M}\\
\leq  2 \bE_{d,f} \big[  \hat\bbP_\tau\big(X(\go) \leq e^{M^2} \big) \big]^{2/3} + 2 e^{-M},
\end{multline}
where we  first used H\"older's inequality, and then \eqref{troctoc}.
The smallness of the first term is the r.h.s\ can be established by using the moment estimates 
from Lemmas \ref{lem:newX2} and  \ref{lem:newvarX2}.
\end{proof}

\subsection{Modifications needed to prove Lemmas \ref{lem:newX2}-\ref{lem:newvarX2}}

\begin{proof}[Proof of Lemma \ref{lem:newvarX2}]
One needs to modify very little of the proof of Lemma \ref{lem:newvarX} in order to obtain Lemma \ref{lem:newvarX2}.
Indeed, one only has to use \eqref{okbathin2} which ensures that $\bbE[\hat\go_i^2]\leq 1+\gep^3/2$ (cf. \eqref{eq:newexpnewvar}). 
Then we notice that the bound \eqref{boundvariance} remains valid, with $2^{q+1}$ replaced by $(1+\gep^3/2)^{q+1}$. Since Lemma~\ref{lestermesenr} also remains valid, we obtain
\begin{multline}
 \bE_{d,f}\left( \Var_{\hat \bbP_{\tau}}\big( X\big) \right) \leq (1+\gep^3/2)^{q+1} + \frac{4^{q+1}}{\ell (D(t))^{q}} \sum_{r=1}^{q} \gb^{2r} c_6 \, q (C_2 D(t))^{q+r-1} \frac{\sqrt{t\ell}}{ \gp(t) \gp(\ell)} \\
 \leq (1+\gep^3/2)^{q+1} + q^2 c_6 \gb^2 (8 C_2^2)^{q+1} \, \frac{1}{D(t) \gp(t) \gp(\ell)}\, \ell^{-\gep^2/2} \, ,
\end{multline}
where, in the last inequality, we used that $D(t)\leq 2/\gb^2$, and that $t\leq  \ell^{1-\gep^2}$. Then, since $q^2$, $(8 C_2^2)^{q+1}$, $D(t)$,  $\gp(t)$, $\gp(\ell)$ are slowly varying functions, the second term goes to $0$ as $\ell$ goes to infinity, and Lemma \ref{lem:newvarX2} is proven.
\end{proof}

\begin{proof}[Proof of Lemma \ref{lem:newX2}]
First of all, one has more refined bounds than \eqref{usefulbounds}: thanks to our choice of $q$ and $t$ in \eqref{def:tq2}, we have that
$\gp(\ell) \leq (e^{\gep^2})^q$, $D(t)\leq (e^{\gep^2})^q$, and  $\gb^2 \geq (1+\gep)/D(t)$. Moreover, if $\gb$ is small enough, one has that $\m_\gb \geq e^{-\gep^2} \gb$. Hence, one can replace \eqref{simplifyEtauX} with 
\begin{equation}
\hat\bbE_\tau [X] 
\geq  \left( \frac{\sqrt{1+\gep} }{ e^{3\gep^2}}\right)^{q+1} \ \times \ \frac{\gp(\ell)}{\sqrt{\ell} \, D(t)^{q}} \sumtwo{\u{i}\in J_{\ell,t}}{ d\le i_0<i_q \le f} U(\u{i})\gd_{\u{i}},
\end{equation}
and Lemma \ref{lem:newX2} follows if one shows that
\begin{equation}
\bP_{d,f}\Bigg( \frac{\gp(\ell)}{\sqrt{\ell} \, D(t)^{q}} \sumtwo{\u{i}\in J_{\ell,t}}{ d\le i_0<i_q \le f} U(\u{i})\gd_{\u{i}} \ \ \geq\ (1-\gep/4)^q \Bigg) \geq 1-e^{-M},
\end{equation}
where we used that for $\gep$ small enough, $$(1+\gep^2)\times \Big(\frac{\sqrt{1+\gep} }{ e^{3\gep^2}} \Big)^{-1} \geq 1-\gep/4.$$ 

Then, the steps \eqref{aaaaaahhhhh}-\eqref{eq:Wamontrer} are identical, and one simply needs to show that the following convergence still holds with our new choice of $t$  (recall  $n_{\ell}=\tfrac14 \eta \ell$).
\begin{equation}
\label{Wconverge}
W_{\ell}:= \frac{\gp(n)}{\sqrt{n} \, D(t)^{q}} \sumtwo{\u{i}\in J'_{n,t}}{ d\le i_0<i_q \le f} U(\u{i})\gd_{\u{i}} \quad \stackrel{\ell \to\infty}{\Longrightarrow} \quad \frac{1}{\sqrt{2\pi}} |Z|, \quad (Z\sim\cN(0,1))\, .
\end{equation}

\smallskip

Thanks to \eqref{convsimple}, one only needs to show that $\bE[\gD W^2]$ converges to $0$ as $\ell \to\infty$, where $\gD W$ is defined in \eqref{def:DW}.
To do so we need a finer control of the covariance terms $\bE[Y_{j_1} Y_{j_2}]$ when $|j_2-j_1| \leq tq$ 
(the definition of $Y_j$ is identical as in \eqref{def:Y}). We prove the following improvement of Lemma~\ref{lem:Y}.
\begin{lemma}
\label{lem:Y2}
We have:
\begin{enumerate}[(i)]
\item If $|j_1-j_2|> tq$, then $\bE[Y_{j_1}Y_{j_2}]=0$ \, ;
\item There exists a constant $C_3>0$ such that, if $j_2\geq j_1$ and $j_2-j_1 \leq tq$, one has
\[\bE[Y_{j_1} Y_{j_2}] \leq (C_3)^q \bE[\gd_{j_1} \gd_{j_2}] = (C_3)^q u(j_1) u(j_2-j_1) \, . \]
\end{enumerate}
\end{lemma}
Thanks to this lemma, one obtains, similarly to \eqref{eq:afterY}
\begin{equation}
\bE[\gD W^2] = \frac{\gp(n)^2}{n} \sumtwo{j_1,j_2 =0}{|j_2-j_1| \leq tq}^n \bE[Y_{j_1} Y_{j_2}] \leq 2 \frac{\gp(n)^2}{n} (C_3)^q \sum_{j_1=0}^n u(j_1)\sum_{j_2=j_1}^{j_1+tq} u(j_2-j_1).
\end{equation}
Then, one uses that $\sum_{i=0}^k u(i) \leq c_8 \gp(k)^{-1} \sqrt{k}$ for all $k\geq 0$, to get that
\begin{equation}
\bE[\gD W^2]\leq 2 \frac{\gp(n)^2}{n} (C_3)^q c_8^2 \frac{\sqrt{n}}{\gp(n)} \frac{\sqrt{tq}}{\gp(tq)} \leq  \frac{2 c_8^2 \sqrt{q} (C_3)^q}{\gp(n) \gp(tq) \sqrt{\eta}}\,  \times \ell^{-\gep^2/4},
\end{equation}
where we used that $t\leq \ell^{1-\gep^2/2}$ and that $n\geq \eta \ell$ to get the last inequality.
One therefore gets that $\bE[\gD W^2]$ converges to $0$ as $n$ goes to infinity, since $\sqrt{q},\gp(n),\gp(tq)$ and $(C_3)^q$ are slowly varying functions.
\end{proof}

\begin{proof}[Proof of Lemma \ref{lem:Y2}]
Conditioning on $\gd_{j_1}$, and denoting $m=j_2-j_1$, one has
\begin{multline}
\bE[Y_{j_1} Y_{j_2}]  = \bE[\gd_{j_1}] \bE\bigg[ \gd_{m} \Big(1-\frac{1}{D(t)^q} \sum_{\u{i}\in J'_{n,t}(0)} U(\u{i}) \gd_{\u{i}} \Big) \Big(1- \frac{1}{D(t)^q} \sum_{\u{k}\in J'_{n,t}(m)} U(\u{k}) \gd_{\u{k}} \Big) \bigg]\\
 \leq \bE[\gd_{j_1}]  \frac{1}{D(t)^{2q}}\bE\bigg[ \sumtwo{\u{i} \in J'_{n,t}(0)}{\u{k} \in J'_{n,t}(m)} U(\u{i}) U(\u{k}) \gd_{\u{i}\,\u{k}}   \bigg] +\bE[\gd_{j_1}] \bE[\gd_m].
\end{multline}
Now, similarly to \eqref{eq:sumUUU}, we prove by induction that there exists a constant $C_3>0$ such that, for any couple $q,q'$ and any $m\geq 0$,
\begin{equation}
\label{eq:sumUUU2}
\sumtwo{\u{i} \in J'_{n,t,q}(0)}{\u{k} \in J'_{n,t,q'}(m)} U(\u{i}) U(\u{k}) U(\u{i}\, \u{k} ) \leq  (C_3)^{q+q'}\, D(t)^{q+q'} \, u(m),
\end{equation}
which is enough to prove Lemma \ref{lem:Y2}.

\medskip
{\bf (0)} The case $q=0$ is trivial: one has $\u{i}=\{0\}$, so that one has the bound
\begin{equation}
\label{eq:q0}
\sum_{\u{k} \in J'_{n,t,q'}(m)} U(\u{k}) U( 0\u{k} )   =  u(m) \sum_{\u{k} \in J'_{n,t,q'}(m)}  U( \u{k} )^2  \leq u(m) D(t)^q \, .
\end{equation}
We now assume that $q\geq 1$.

\medskip
{\bf (1)} We first show by induction on $q$ that there exists a constant $C_4$, such that when $q'=0$ (so that $\u{k}= \{m\}$),
\begin{equation}
\label{eq:q'0}
\sum_{\u{i} \in J'_{n,t,q}(0)} U(\u{i}) U(\u{i} \, m) \leq  (C_4)^{q} D(t)^q u(m).
\end{equation}
We decompose the sum according to whether $i_{q-1}\geq m$, $i_{q-1} < m \leq i_q$ or $i_{q-1}<i_q<m$.

{\bf a.} If $i_{q-1}\geq m $, one trivially has that $\sum_{i_q=i_{q-1}+1}^{i_{q-1}+t} u(i_q - i_{q-1})^2 \leq D(t)$, so that
\begin{equation}
\label{eq:sum1}
\sumtwo{\u{i} \in J'_{n,t,q}(0)}{i_{q-1} \geq m} U(\u{i}) U(\u{i}\, m) \leq   D(t) \sumtwo{\u{i} \in J'_{n,t,q-1}(0)}{i_{q-1} \geq m} U(\u{i}) U(\u{i}\, m)  .
\end{equation}

{\bf b.} If $i_{q-1} < m \leq i_{q}$, then using that $u(i_q-i_{q-1})\leq c_5 u(i_q - m)$, one gets that
\begin{equation}
\sum_{i_q=m}^{i_{q-1}+t} u(i_{q}-m) u(m-i_{q-1}) u(i_q-i_{q-1}) \leq c_5 D(t) u(m-{i_{q-1}}),
\end{equation}
so that
\begin{equation}
\label{eq:sum2}
\sumtwo{\u{i} \in J'_{n,t,q}(0)}{i_{q-1} < m \leq i_q} U(\u{i}) U(\u{i}\, m) \leq  c_5 D(t) \sumtwo{\u{i} \in J'_{n,t,q-1}(0)}{i_{q-1} < m} U(\u{i}) U(\u{i}\, m)  .
\end{equation}

{\bf c.} Now, if $ i_{q-1}<i_q <m$, we have that there exists a constant $c_9$ such that
\begin{equation}
\sum_{i_q=i_{q-1}+1}^{\min(i_{q-1}+t, m)} u(m - i_{q}) u(i_q-i_{q-1})^2 \leq c_9 D(t)\, u(m- i_{q-1}).
\end{equation}
Indeed, if we denote $x=m-i_{q-1}$, we can decompose the above sum into whether $i_q \leq i_{q-1} + x/2$ or $i_q> i_{q-1} + x/2$.
If $i_q \leq i_{q-1} + x/2$, then $m - i_{q} \geq \frac12 (m-i_{q-1})$, so that $u(m - i_{q}) \leq c_{7} u(m-i_{q-1})$ ($c_7$ is a constant such that $u(m)\leq c_7 u(n)$ whenever $m\geq \tfrac14 n$).
Then
\[
\sum_{i_q=i_{q-1}+1}^{\min(i_{q-1}+t, i_{q-1}+x/2)} u(m - i_{q}) u(i_q-i_{q-1})^2  \leq c_{7} u(m-i_{q-1}) \sum_{i_q=i_{q-1}+1}^{i_{q-1}+t}  u(i_q-i_{q-1})^2 \leq c_{7} u(m-i_{q-1}) D(t) .
\]
On the other hand, if $i_{q} \geq i_{q-1} + x/2$, then $u(i_q-i_{q-1}) \leq c_{7} u(m-i_{q-1})$, and since $m-i_{q} \leq x/2 \leq i_{q}-i_{q-1}$, one also has that $u(m-i_q) \leq c_5 u(i_{q}-i_{q-1})$. One then has that
\[
\sum_{i_q=i_{q-1}+x/2}^{\min(i_{q-1}+t,m)}u(m - i_{q}) u(i_q-i_{q-1})^2  \leq c_{7} c_5 u(m-i_{q-1}) \sum_{i_q= i_{q-1}+x/2}^{\min(i_{q-1}+t,m)}  u(i_q-i_{q-1})^2 \leq c_{7} c_5  u(m-i_{q-1}) D(t).
\]
Hence, we showed that
\begin{equation}
\label{eq:sum3}
\sumtwo{\u{i} \in J'_{n,t,q}(0)}{i_{q-1} <i_q< m } U(\u{i}) U(\u{i} \, m) \leq  (c_{7}+ c_{7} c_5) D(t) \sumtwo{\u{i} \in J'_{n,t,q-1}(0)}{i_{q-1} < m} U(\u{i}) U(\u{i}\, m) ,
\end{equation}
which is \eqref{eq:sum3} with $c_9 = c_{7}+ c_{7} c_5$.

Combining \eqref{eq:sum1}-\eqref{eq:sum2}-\eqref{eq:sum3}, and setting $C_4:= \max(c_5,c_9)$, we have that
\begin{equation}
\sum_{\u{i} \in J'_{n,t,q}(0)} U(\u{i}) U(\u{i} \, m) \leq  C_4 D(t) \sum_{\u{i} \in J'_{n,t,q-1}(0)} U(\u{i}) U(\u{i}\, m),
\end{equation}
which by iteration gives \eqref{eq:q'0}.

\medskip
{\bf (2)} Then, assume that $q'\geq 1$. As in the proof of \eqref{eq:sumUUU}, we decompose the sum into two components, according ot whether $i_q$ or $k_{q'}$ is larger, and one obtains, exactly as in \eqref{eq:lastsumUUU},
\begin{equation}
\sumtwo{\u{i}  \in J'_{n,t,q}(0)}{\u{k}  \in J'_{n,t,q'}(0)} \!\! U(\u{i}) U(\u{k})U(\u{i}\, \u{k})\le  2c_5 \max\Bigg\{  \sumtwo{\u{i}  \in J'_{n,t,q-1}(0)}{\u{k}  \in J'_{n,t,q'}(0)}\!\!  U(\u{i}) U(\u{k})U(\u{i}\, \u{k}) \ ; \!\!\!\!  \sumtwo{\u{i}  \in J'_{n,t,q}(0)}{\u{k}  \in J'_{n,t,q'-1}(0)} \!\! U(\u{i}) U(\u{k})U(\u{i}\, \u{k})\Bigg\}\, ,
\end{equation}
which in turns gives \eqref{eq:sumUUU2} by induction, thanks to \eqref{eq:q0}-\eqref{eq:q'0}, with  $C_3:=\max(2 c_5, C_4)$.
\end{proof}

\bibliographystyle{plain}

\begin{thebibliography}{99}
\bibitem{cf:A06}  {K.~S.~Alexander}, {\sl The effect of disorder on
    polymer depinning transitions}, Commun. Math. Phys. {\bf 279} (2008),
     117-146.

\bibitem{cf:AS} K. Alexander and V. Sidoravicius \emph{
Pinning of polymers and interfaces by random potentials}
    Ann. Appl. Probab.
  {\bf 16}, (2006) 636-669.


  \bibitem{cf:AZ08} {K.~S.~Alexander and N.~Zygouras}, {\sl Quenched and
    annealed critical points in polymer pinning models}, Comm. Math.
  Phys. {\bf 291} (2009), 659-689.
  

  \bibitem{cf:AZ10}  {K.~S.~Alexander and N.~Zygouras} \emph{
Equality of Critical Points for Polymer Depinning Transitions with
Loop Exponent One}, Ann. Appl. Prob
{\bf 20}
(2010), 356-366.

  \bibitem{cf:AZ12} K.~S.~Alexander and N.~Zygouras  \emph{
Path properties of the disordered pinning model in the delocalized regime},
Annals of Applied Probability {\bf 24} (2014) 599-615.

\bibitem{cf:QB}
   Q.~Berger, 
\emph{Pinning model in random correlated environment: appearance of an infinite disorder regime},  
J. Stat. Phys. {\bf 155} (2014),  544-570.   
  
  
 \bibitem{cf:Magic5} Q.~Berger, F.~Caravenna, J.~Poisat, R.~Sun and N.~Zygouras 
  \emph{ The Critical Curve of the Random Pinning and Copolymer Models at Weak Coupling}, 
  Commun. Math. Phys. {\bf 326} (2014), 507-530. 

 \bibitem{cf:QH}
 Q.~Berger
 and H.~Lacoin, \emph{Sharp critical behavior for pinning models in a random correlated environment}, Stochastic Process. Appl. {\bf 122} (2012),  1397-1436.


\bibitem{cf:QF1}
Q.~Berger and F.~L.~Toninelli \emph{On the Critical Point of the Random Walk Pinning Model in Dimension d=3},
          Elec. J. Probab.
    {\bf 15} (2010), 654, 683.

\bibitem{cf:QF}
Q.~Berger and F.~L.~Toninelli \emph{Hierarchical pinning model in correlated random environment},
          Ann. Inst. H. Poincaré Probab. Statist.
    {\bf 49} (2013), 781-816.

\bibitem{cf:BM1}  S.~M.~Bhattacharjee and S.~Mukherji, \emph{Directed polymers with random interaction - Marginal relevance and novel criticality}, Phys. Rev. Lett. {\bf 70} (1993), 49-52.

\bibitem{cf:BM2} S.~M.~Bhattacharjee and S.~Mukherji, \emph{Directed polymers with random interaction - An exactly solvable case}, Phys. Rev. E, {\bf 48} (1993), 3483-3496.    
    
    
\bibitem{cf:BS1} M.~Birkner and R.~Sun \emph{Annealed vs quenched critical points for a random walk pinning model} Ann. Inst. H. Poincaré Probab. Statist. {\bf 46} (2010), 414-441.  
    
\bibitem{cf:BS2} M.~Birkner and R.~Sun \emph{Disorder relevance for the random walk pinning model in dimension 3} Ann. Inst. H. Poincaré Probab. Statist. {\bf 47} (2011), 259-293.

\bibitem{cf:BGLT} T.~Bodineau, G.~Giacomin, H.~Lacoin and F.~L.~Toninelli, \emph{Copolymers at selective interfaces: new bounds on the phase diagram}, J. Statist. Phys. {\bf132} (2008) 603-626 

\bibitem{cf:BGT} N.~H.~Bingham, C.~M.~Goldie and J.~L.~Teugels, {\sl Regular Variations}, Cambridge University Press (1987).

\bibitem{cf:CdH} F.~Caravenna and F.~den Hollander
 \emph{A general smoothing inequality for disordered polymers} Electron. Commun. Probab. {\bf 18} (2013), no. 76, 1-15. 
 
\bibitem{cf:CSZ1} F.~Caravenna, R.~Sun and N.~Zygouras \emph{Polynomial chaos and scaling limits of disordered systems} to appear in J. Eur. Math. Soc. (JEMS).

\bibitem{cf:CSZ2}   F.~Caravenna, R.~Sun and N.~Zygouras 
  \emph{ The continuum disordered pinning model} to appear in Prob. Theor. Relat. Fields.
 
\bibitem{cf:CSZ3}   F.~Caravenna, R.~Sun and N.~Zygouras 
  \emph{The universal scaling limit of marginally relevant disordered systems} (in preparation)
 
\bibitem{cf:CTT} F.~Caravenna, F.~L.~Toninelli and N. Torri
\emph{Universality feature of the Pinning Model} (in preparation) 
 
 \bibitem{cf:CDH} D.~Cheliotis and F.~den Hollander
\emph{Variational characterization of the critical curve for pinning of random
polymers}
Ann. Probab. {\bf 41} (2013), 1767-1805. 
  
 \bibitem{cf:CH} D.~Cule and T.~Hwa,
\emph{Denaturation of Heterogeneous DNA},
Phys. Rev. Lett. {\bf 79} (1997), 2375  
  
\bibitem{cf:DHV} B.~Derrida, V.~Hakim and J.~Vannimenus,
  {\sl Effect of disorder on two-dimensional wetting}, J. Statist.
  Phys. {\bf 66} (1992), 1189-1213.

\bibitem{cf:DGLT09} B.~Derrida, G.~Giacomin, H.~Lacoin and F.~.L.~Toninelli,
  {\sl Fractional moment bounds and disorder relevance for pinning models}, Commun. Math. Phys. {\bf 287} (2009), 867-887.

\bibitem{cf:DR}
B.~Derrida and M.~Retaux, \emph{The depinning transition in presence of disorder: a toy model}, J. Stat. Phys. {\bf 156} (2014), 268-290

\bibitem{cf:Doney} R.~A.~Doney, {\sl One-sided local large deviation and renewal theorems in the case of infinite mean}, Probab. Theory Relat. Fields, {\bf 107} (1997), 451-465.
  
  
\bibitem{cf:Fisher}
M.~E.~Fisher, \emph{Walks, walls, wetting, and melting},
J. Statist. Phys. {\bf 34} (1984), 667-729.

\bibitem{cf:FLNO}
G.~Forgacs, J.~M.~Luck, Th.~M.~Nieuwenhuizen and H.~Orland,
\emph{Wetting of a disordered substrate: Exact critical behavior in two dimensions}, Phys. Rev. Lett. {\bf 57} (1986), 2184-2187.


\bibitem{cf:GN}
D.~M.~Gangardt and S.~K.~Nechaev, \emph{Wetting transition on a one-dimensional disorder}, J. Stat.
Phys. 130 (2008), 483-502.


\bibitem{cf:GB} G. Giacomin, {\sl Random polymer models}, 
Imperial College Press, World Scientific (2007). 

\bibitem{cf:G} G. Giacomin, \emph{Disorder and critical phenomena through basic probability models}, \'Ecole d'\'et\'e de probablit\'es de Saint-Flour XL-2010, Lecture Notes in Mathematics {\bf 2025}, Springer, (2011).

\bibitem{cf:GL15} G.~Giacomin and H.~Lacoin \emph{Pinning and disorder relevance for the discrete Gaussian free-field} (preprint, 	arXiv:1501.07909 [math.PR]).


\bibitem{cf:GLT} G.~Giacomin, H.~Lacoin and F.~L.~Toninelli, \emph{
    Hierarchical pinning models, quadratic maps and quenched
    disorder}, Probab. Theor. Rel. Fields {\bf 147} (2010) 185-216.

  \bibitem{cf:GLT10} {G.~Giacomin, H.~Lacoin and F.~L.~Toninelli},
{\sl Marginal relevance of disorder for pinning models},
Commun. Pure Appl. Math. {\bf 63} (2010) 233-265. 

\bibitem{cf:GLT11} {G.~Giacomin, H.~Lacoin and F.~L.~Toninelli},
  \textit{Disorder relevance at marginality and critical point shift}   Ann. Inst. H. Poincaré {\bf 47} (2011) 148-175.

 
 \bibitem{cf:GT_cmp} G.~Giacomin and  F.~L.~Toninelli, \emph{Smoothing effect of quenched disorder on polymer depinning transitions},
Commun. Math. Phys. {\bf 266} (2006), 1--16.


\bibitem{cf:GT09} G.~Giacomin and F.~L.~Toninelli \emph{On the irrelevant disorder regime of
pinning models.}
Ann. Probab.
{\bf 37} (2009).
1841–1875.




 \bibitem{cf:Hcrit} A.~B.~Harris, {\it Effect of random defects on the critical behaviour of Ising models} J. Phys. C {\bf 7} (1974),
1671–1692.

\bibitem{cf:KM}
Y.~Kafri, D.~Mukamel, \emph{Griffiths singularities in unbinding of strongly disordered polymers}, Phys. Rev. Lett. {\bf 91} (2003), 038103.

\bibitem{cf:KL}
H.~Kunz, R.~Livi, \emph{DNA denaturation and wetting in the presence of disorder}, Eur. Phys. Lett. {\bf 99} (2012), 30001.

\bibitem{cf:Lhier} H.~Lacoin, {\it Hierarchical pinning model with site disorder: Disorder is marginally relevant}, 
Probab. Theor. Relat. Field.
{\bf 148} (2010) 159-175

\bibitem{cf:L} H.~Lacoin, {\it New bounds for the free energy of directed polymers in dimension $1+1$ and $1+2$}, 
Commun. Math. Phys.
{\bf 294} (2010) 471-503. 

\bibitem{cf:Lmart} H.~Lacoin,  \emph{The martingale approach to disorder irrelevance for pinning models}, Elec. Comm. Probab. {\bf 15} (2010) 418-427.

\bibitem{cf:LSAW} H.~Lacoin, \emph{Non-coincidence of Quenched and Annealed Connective Constants on the supercritical planar percolation cluster} Probability Theory and Related Fields {\bf 159} (2014) 777-808.

\bibitem{cf:lighttail} H.~Lacoin, \emph{The rounding of the phase transition for disordered pinning with stretched exponential tails}, (preprint) arXiv:1405.6875 [math-ph]

\bibitem{cf:Monthus}
C.~Monthus, \emph{Random Walks and Polymers in the Presence of Quenched Disorder} Lett. Math. Phys. {\bf 78} (2006), 207-233.

\bibitem{cf:MG}
C.~Monthus, T.~Garel, \emph{Multifractal statistics of the local order parameter at random critical points: Application to wetting transitions with disorder}, Phys. Rev. E {\bf 76} (2007), 021114.


  \bibitem{cf:Nak}
  M.~Nakashima, \emph{
A remark on the bound for the free energy of directed polymers in random environment in 1+2 dimension} (preprint)
 	arXiv:1406.4614 [math.PR]
 
\bibitem{cf:Poi} J.~Poisat \emph{Random pinning model with finite range correlations : disorder relevant regime}
 Stochastic Processes and their Applications {\bf 122} (2012) 3560-3579.   

\bibitem{cf:TC}
L.~H.~Tang and H.~Chat\'e, \emph{Rare-event induced binding transition of heteropolymers}, Phys. Rev. Lett. {\bf 86} (2001), 830-833.

\bibitem{cf:T08}  F.~L.~Toninelli,
\emph{A replica-coupling approach to disordered pinning models},
Commun. Math. Phys. {\bf 280} (2008), 389-401.

\bibitem{cf:T_fractmom}
F.~L.~Toninelli, \emph{Disordered pinning models and
    copolymers: beyond annealed bounds},  Ann. Appl.
  Probab. {\bf 18} (2008), 1569-1587.
  
\bibitem{cf:Z} N.~Zygouras, \emph{Strong disorder in semidirected random polymers},   Ann. Inst. H. Poincaré 
Probab. Statist. {\bf 49}  (2013), 753-780.

\bibitem{cf:YZ} A.~Yilmaz and  O.~Zeitouni, \emph{Differing Averaged and Quenched Large Deviations for Random Walks in Random Environments in Dimensions Two and Three}, Commun. Math. Phys. {\bf 300} (2010) 243-271. 

\end{thebibliography}

\end{document}